\documentclass[11pt,letterpaper]{scrartcl}
\usepackage{hyperref}
\usepackage{enumerate}
\usepackage{amsthm}
\usepackage{amsmath}
\usepackage{amssymb}
\usepackage{mathrsfs}
\usepackage{mathtools}
\usepackage{verbatim}
\usepackage[capitalise]{cleveref}
\usepackage{relsize}
\usepackage{xspace}
\usepackage{nicematrix}
\usepackage{fullpage}
\usepackage{float}

\allowdisplaybreaks

\theoremstyle{definition}
\newtheorem{definition}{Definition}[section]
\theoremstyle{plain}
\newtheorem{lemma}[definition]{Lemma}
\newtheorem{theorem}[definition]{Theorem}
\newtheorem{corollary}[definition]{Corollary}
\theoremstyle{remark}
\newtheorem{remark}{Remark}
\newtheorem{observation}[definition]{Observation}

\usepackage{macros}
\title{A tight Monte-Carlo algorithm for Steiner Tree parameterized by clique-width}
\author{Narek Bojikian\hspace{2cm}Stefan Kratsch}

\begin{document}

\maketitle

\thispagestyle{empty}

\begin{abstract}
 \small

 \noindent Recently, Hegerfeld and Kratsch [ESA 2023] obtained the first tight algorithmic results for hard connectivity problems parameterized by clique-width. Concretely, they gave one-sided error Monte-Carlo algorithms that given a $k$-clique-expression solve \textsc{Connected Vertex Cover} in time $6^kn^{\Oh(1)}$ and \textsc{Connected Dominating Set} in time $5^kn^{\Oh(1)}$. Moreover, under the Strong Exponential-Time Hypothesis (SETH) these results were showed to be tight. Their work builds on a research program of determining tight complexity bounds that was initiated by work of Lokshtanov et al.~[SODA 2011 \& TALG 2018], about the first tight lower bounds relative to treewidth (modulo SETH), and work of Cygan et al.~[FOCS 2011 \& TALG 2022], about the cut-and-count framework and tight bounds for connectivity problems parameterized by treewidth.
 
 Hegerfeld and Kratsch [2023] leave open, however, several important benchmark problems, whose complexity relative to treewidth had been settled by Cygan et al., namely \textsc{Steiner Tree}, \textsc{Connected Odd Cycle Transversal}, and \textsc{(Connected) Feedback Vertex Set}. As a key obstruction they point out the exponential gap between the rank of certain compatibility matrices, which is often used for algorithms, and the largest triangular submatrices therein, which is essential for current lower bound methods. E.g., for \textsc{Steiner Tree} parameterized by clique-width the GF(2)-rank is at least $4^k$, while no triangular submatrix larger than $3^k\times 3^k$ was known. This at best yields an upper bound of time $4^kn^{\Oh(1)}$, while the obtainable lower bound of time $(3-\varepsilon)^kn^{\Oh(1)}$ under SETH was already known relative to pathwidth.
 
 We close this gap by showing that, somewhat surprisingly, \Stp\ can be solved in time $3^kn^{\Oh(1)}$ when given a $k$-clique-expression. Hence, for all parameters between cutwidth and clique-width it has the same tight complexity. We first show that there is a ``representative submatrix'' of GF(2)-rank $3^k$ (also ruling out larger than $3^k\times 3^k$ triangular submatrices). In principle, this could be used for an algorithm via the representative sets-based (or rank-based) approach of Bodlaender et al.~[ICALP 2013 \& IANDC 2015] but such an algorithm would not be sufficiently fast. It is tempting to additionally leverage the ``low'' GF(2)-rank of the submatrix but, at first glance, this leaves us with a useless combination of algorithmic approaches: Representative sets-based algorithms are slower and preserve \emph{existence} of a solution, whereas leveraging ``low'' GF(2)-rank, e.g., by cut-and-count, usually gives faster algorithms but preserves the \emph{parity} of the number of solutions. An outright combination of the two approaches cannot be expected to always yield the correct outcome because representation may change the parity of solutions (even when assuming isolation of a single solution). We nevertheless show how to reconcile these two worlds, thereby obtaining the claimed time complexity. We believe that our approach will be instrumental for settling further open problems in this research program.
\end{abstract}
\setcounter{page}{0}
\newpage

\section{Introduction}

Leveraging input structure to get faster algorithms is a fundamental strategy for coping with \classNP-hardness and other forms of intractability. This quickly leads to asking \emph{``how does structure affect complexity?''} This question lies at the heart of \emph{parameterized complexity}, where one quantifies input structure through so-called \emph{parameters} and seeks to determine the complexity of the corresponding \emph{parameterized problems} as a function of input size $n$ \emph{and} parameter value $k$. Initially, this has led to the classification of many hard problems as either being solvable in time $f(k)\cdot n^{\Oh(1)}$ for some function $f$, i.e., being \emph{fixed-parameter tractable} (\classFPT), or as being \classW{1}-hard and thereby unlikely to be \classFPT. Later, it became possible to prove tight bounds on the parameter dependence $f(k)$ for many problems subject to the (Strong) Exponential-Time Hypothesis (ETH/SETH).\footnote{ETH is the hypothesis that there is $c>1$ such that \textsc{3-CNF SAT} with $n$ variables cannot be solved in time $\Oh(c^n)$, while SETH is the hypothesis that for each $c<2$ there is $q\in\naturals$ such that \textsc{$q$-CNF SAT} cannot be solved in time $\Oh(c^n)$. It is known that SETH implies ETH, which in turn implies $\classFPT\neq\classW{1}$ and $\classP\neq\classNP$.} Typically, for a problem solvable in time $\alpha^k\cdot n^{\Oh(1)}$ one may be able to rule out time $2^{o(k)}\cdot n^{\Oh(1)}$ assuming ETH, while for some problems even time $(\alpha-\varepsilon)^k\cdot n^{\Oh(1)}$ can be ruled out for all $\varepsilon>0$ assuming SETH. In the following, we focus on the latter type of tight lower bounds, i.e., getting exact bases~$\alpha$.

This endeavor has been most successful for problems parameterized by \emph{treewidth},\footnote{Intuitively, the treewidth of a graph $G$, denoted $\tw(G)$, is the smallest value $k$ such that $G$ can be completely decomposed by non-crossing vertex separators of size at most $k$ each. Two important obstructions are cliques and grids as minors of $G$. Generally, only sparse graphs may have small treewidth as number of edges $m\leq \tw(G)\cdot n$.} a ubiquitous measure of graphs not just in parameterized complexity. The program of determining exact complexity was initiated by Lokshtanov et al.~\cite{DBLP:conf/soda/LokshtanovMS11a,DBLP:journals/talg/LokshtanovMS18} who showed that many well-known algorithms for problems parameterized by treewidth are essentially optimal modulo SETH, e.g., \textsc{Vertex Cover[$\tw$]}\footnote{We often write \textsc{Foo[$bar$]} to mean the parameterization of problem \textsc{Foo} by parameter $bar$.} can be solved in time $2^{tw}\cdot n^{\Oh(1)}$ but not in time $(2-\varepsilon)^{tw}\cdot n^{\Oh(1)}$, for any $\varepsilon>0$, assuming SETH. Another breakthrough was obtained by Cygan et al.~\cite{DBLP:journals/talg/CyganNPPRW22} who used their novel \emph{cut-and-count} technique to show, surprisingly, that many important connectivity-related problems have single-exponential (randomized) algorithms too. Moreover, they showed the obtained bounds to be optimal modulo SETH, e.g., \textsc{Connected Vertex Cover[$\tw$]} and \textsc{Steiner Tree[$\tw$]} can be solved in time $3^{tw}\cdot n^{\Oh(1)}$ but not in time $(3-\varepsilon)^{tw}\cdot n^{\Oh(1)}$ for any $\varepsilon>0$, assuming SETH. Since then, there has been much progress on tight bounds relative to treewidth (and the closely related \emph{pathwidth}), see, e.g.,~\cite{DBLP:conf/iwpec/BorradaileL16,DBLP:conf/soda/CurticapeanM16,DBLP:conf/soda/CurticapeanLN18,DBLP:journals/jacm/CyganKN18,DBLP:conf/esa/OkrasaPR20,DBLP:journals/tcs/HanakaKS21,DBLP:conf/icalp/MarxSS21,DBLP:journals/siamcomp/OkrasaR21,DBLP:journals/corr/abs-2210-10677/EsmerFMR22,DBLP:conf/soda/FockeMR22,DBLP:journals/dam/KatsikarelisLP22,DBLP:conf/soda/FockeMINSSW23}.

While treewidth and pathwidth quantify how well a graph can be decomposed along (small) vertex separators, it is natural to ask for decomposition along (small) edge cuts, which is more restrictive and may permit faster algorithms for well-decomposable graphs. This leads us to \emph{cutwidth} as an analogue of pathwidth (along with several tree-like variants, e.g.,~\cite{DBLP:conf/mfcs/GanianKS15,DBLP:conf/wg/BrandCGHK22,DBLP:conf/iwpec/GanianK22}). There are a number of tight bounds known relative to cutwidth~\cite{DBLP:journals/jgaa/GeffenJKM20,DBLP:journals/tcs/JansenN19,DBLP:conf/icalp/MarxSS21,DBLP:conf/stacs/PiecykR21,DBLP:conf/stacs/GroenlandMNS22,DBLP:conf/stacs/BojikianCHK23}, including tight bounds for
\textsc{Connected Vertex Cover}, \textsc{Connected Dominating Set}, \textsc{Feedback Vertex Set}, \textsc{Steiner Tree}, and \textsc{Connected Odd Cycle Transversal}~\cite{DBLP:conf/stacs/BojikianCHK23} as well as for counting connected edge sets~\cite{DBLP:conf/stacs/GroenlandMNS22}.
As treewidth (and more restrictive parameters like cutwidth and its variants) can only be small on sparse graphs, however, we need to look towards more general structure/parameters to understand the complexity of problems relative to structure present in dense/general graphs.

For this goal, \emph{clique-width} and \emph{rank-width} clearly stand out as canonical targets. The former measures how well the graph can be constructed when adjacency is controlled via $k$ labels/colors; the latter measures how well the graph can be decomposed by non-crossing cuts of small $\bin$-rank. Recently, Bergougnoux et al.~\cite{DBLP:conf/stacs/BergougnouxKN23} gave the first lower bounds for problems parameterized by rank-width, assuming ETH, but we are not aware of any results with tight bases even under SETH.\footnote{It should be noted that typical upper bounds relative to rank-width are not analyzed to exact bases and constants in the exponents, nor do the algorithms seem likely to be optimal. Such tight bounds may still be far off.} Complementing this, there is by now a handful of tight complexity bounds relative to clique-width (modulo SETH)~\cite{DBLP:conf/esa/IwataY15,DBLP:conf/soda/CurticapeanM16,DBLP:journals/dam/KatsikarelisLP19,DBLP:journals/siamdm/Lampis20,DBLP:conf/icalp/GanianHKOS22,DBLP:conf/esa/HegerfeldK23}. 
However, among these tight bounds there are almost none about connectivity-related problems (or problems with other kinds of non-local constraints), which are a natural target for theory development, especially when compared to what we know relative to treewidth/pathwidth as well as cutwidth. Only very recently did Hegerfeld and Kratsch~\cite{DBLP:conf/esa/HegerfeldK23} show how to solve \textsc{Connected Vertex Cover[$\cw$]} in time $6^{\cw}\cdot n^{\Oh(1)}$ and \textsc{Connected Dominating Set[$\cw$]} in time $5^{\cw}\cdot n^{\Oh(1)}$; moreover, they proved matching lower bounds modulo SETH. They conclude, however, that using their techniques they are unable to settle the complexity of other benchmark problems for which tight bounds are known relative to other width parameters, among them \textsc{Steiner Tree}, \textsc{Connected Odd Cycle Transversal}, and \textsc{Feedback Vertex Set}. What are these techniques and what seems to be the obstruction?

\paragraph{State of the art.}
Dynamic programming (DP) is the dominant (if not sole) algorithmic paradigm for dealing with problems on graphs of small width, i.e., problems parameterized by treewidth, cutwidth, clique-width, etc. The corresponding decomposition of the graph is traversed in a bottom-up manner while managing a sufficiently large selection of partial solutions or just some fingerprint thereof. Naive DP often relies on understanding the different ways in which a partial solution can interact with the rest of the graph (at cut or separator that is small or structured, depending on the used width parameter), and maintaining information for each type of interaction: 
Existence, optimal cost/value, or number of such partial solutions. 
Advanced DP, such as the cut-and-count technique~\cite{DBLP:journals/talg/CyganNPPRW22} or the so-called \emph{rank-based} approach~\cite{DBLP:journals/iandc/BodlaenderCKN15}, often relies on properties of certain \emph{matrices} that are already implicit in naive DP:
Rows of the matrix correspond to the possible interactions of partial solutions for a subgraph (for part of the decomposition); columns correspond to the ways of completing a partial solution to a solution for the full graph (not all of these need to be possible for any given graph). The entries of the matrix capture which combinations are possible. In this way, the dimensions of these matrices are usually equal to the number of types of partial solutions that are considered in a naive DP. Crucially, however, such a matrix may have favorable properties such as low rank (or even a useful low-rank factorization) or at least absence of large triangular/permutation submatrices, i.e., these may be much smaller than implied by its dimensions. (Now seems a good time to direct the reader to the insightful survey by Nederlof~\cite{DBLP:conf/birthday/Nederlof20}.)

So far, specific matrix properties seem aligned with the different problem types: Complexity of counting solutions over a field \fieldF depends on the \fieldF-rank of the matrix, complexity of finding an optimum cost/value solution (for small weights) depends on the maximum size of triangular submatrices, and complexity of decision depends on the maximum size of permutation submatrices.\footnote{A recent survey talk by Jesper Nederlof at Lorentz Center in Leiden (NL) helped (re)enforce this perspective.} Fortunately for settling complexity of optimization/decision problems, in most cases rank over some sensible field and dimensions of largest triangular/permutation submatrices coincide.\footnote{Overwhelmingly, such as with cut-and-count, one uses GF(2)-rank, which, however, corresponds to counting solutions modulo two. The well-known Isolation Lemma~\cite{MulmuleyVV87,DBLP:journals/talg/CyganNPPRW22} allows the necessary reduction of (small weight) optimization/decision to counting solutions of specified small weight modulo two.} This often enables the use of low-rank factorizations, i.e., intuitively a transformation of the space of partial solutions to a more convenient (i.e., smaller) one, with dimension then matching the rank, leading to the required fast DP algorithm. Apart from this, we still lack general methods for leveraging directly the absence of large triangular/permutation submatrices when the rank is larger. In principle, triangular/permutation submatrices are strongly related to bounds for so-called \emph{representative (sub)sets} of partial solutions,\footnote{Roughly, this means that having any set $\calS$ of partial solutions, there is a ``small'' \emph{representative} subset $\calS'\subseteq\calS$ that permits the same extensions to complete solutions, though not the same number of complete solutions.} so upper bounds on the dimensions of such submatrices are upper bounds on the required size of representative sets. Several results for (non-tight) DP algorithms rely on representative sets~\cite{DBLP:journals/iandc/BodlaenderCKN15,DBLP:journals/jacm/FominLPS16,DBLP:journals/talg/FominLPS17,DBLP:journals/tcs/BergougnouxK19,DBLP:journals/siamdm/BergougnouxK21}, often building on the \emph{rank}-based approach of Bodlaender et al.~\cite{DBLP:journals/iandc/BodlaenderCKN15}, so, not at all surprisingly, upper bounds on the rank are the means by which one leverages absence of the large triangular/permutation submatrices.\footnote{On the positive side, this works over any field, solves optimization with arbitrary weights, and works without randomization, such as required for the Isolation Lemma.}
For lower bounds, the story so far seems simpler, by comparison at least. Since the work of Lokshtanov et al.~\cite{DBLP:conf/soda/LokshtanovMS11a,DBLP:journals/talg/LokshtanovMS18} and Cygan et al.~\cite{DBLP:journals/corr/abs-1103-0534/CyganNPPRW11,DBLP:journals/talg/CyganNPPRW22}, lower bounds have certainly become increasingly complicated but the underlying principles largely remain the same. In particular, with current techniques there is no reason to hope for higher lower bounds than what we can get via the corresponding matrix property, e.g., not more than dimensions of largest triangular submatrix for optimization problems.\footnote{Notably, the situation for counting problems over some field $\fieldF$ seems easier as complexity appears to depend directly on the \fieldF-rank. That being said, crafting lower bound gadgets for these may be harder (see, e.g.,~\cite{DBLP:conf/soda/CurticapeanM16,DBLP:conf/soda/CurticapeanLN18,DBLP:conf/stacs/GroenlandMNS22,DBLP:conf/icalp/MarxSS21,DBLP:conf/soda/FockeMINSSW23}).}

\paragraph{Steiner Tree.}
Returning to \textsc{Steiner Tree[$\cw$]}, the $GF(2)$-rank of the relevant matrix, using graphs with $k=\cw$ labels, turns out to be (at least) $4^k$ while the largest known triangular submatrix is $3^k\times 3^k$ (basically inherited from \textsc{Steiner Tree} parameterized by treewidth or pathwidth). Current methods (such as in~\cite{DBLP:conf/esa/HegerfeldK23}) therefore yield an upper bound of time $4^{cw}\cdot n^{\Oh(1)}$ while there is no $(3-\varepsilon)^{\cw}\cdot n^{\Oh(1)}$ time algorithm, for any $\varepsilon>0$, assuming SETH. The most convenient solution would be to identify a $4^k\times 4^k$ triangular submatrix as this would likely lead to ruling out time $(4-\varepsilon)^{\cw}\cdot n^{\Oh(1)}$ (spoiler: this does not exist). Otherwise, one needs to leverage the absence of larger triangular submatrices but without the usual help of having a matching rank bound. Using representative sets (like in the rank-based approach) is unlikely to give a tight bound, due to the cost of reducing partial solutions to representative subsets via Gaussian elimination. We are aware of only two works so far that manage to overcome such a gap between rank and maximum size of relevant submatrices and obtain tight complexity bounds; both are for problems parameterized by cutwidth~\cite{DBLP:journals/tcs/JansenN19,DBLP:conf/stacs/BojikianCHK23}. Their approaches do not seem to transfer due to dependence on the specific problems and, crucially, on properties of graphs of small cutwidth.

\begin{table}
\newcommand{\expk}[1]{$\Oh^*({#1}^k)$}%
\newcommand{\tablespacing}{\hspace{0.25cm}}
\centering
\begin{tabular}{l@{\tablespacing}|@{\tablespacing}c@{\tablespacing}c@{\tablespacing}c@{\tablespacing}c}%
     & cutwidth & treewidth & modular-tw & clique-width\\%
    \hline%
    \\[-1em]
    \textsc{\textsc{$q$-Coloring}} & \expk{2} & \expk{q} & \expk{\binom{q}{\lfloor q/2\rfloor}} & \expk{(2^q-2)} \\%
    \textsc{Vertex Cover} & \expk{2} & \expk{2} & \expk{2} & \expk{2} \\
    \textsc{Connected Vertex Cover} & \expk{2} & \expk{3} & \expk{5} & \expk{6} \\
    \textsc{Connected Dominating Set} & \expk{3} & \expk{4}	& \expk{4} & \expk{5} \\
    \textsc{Steiner Tree}  & \expk{3} & \expk{3} & \expk{3} & $\boldsymbol{\ostar(3^k)}$ \\
    \textsc{Feedback Vertex Set} & \expk{2} & \expk{3} & \expk{5} & ? \\
    \textsc{Connected Odd Cycle Transversal} & \expk{4} & \expk{4} & ? & ? \\
    \hline%
    \\[-1em]
    References & \cite{DBLP:conf/stacs/BojikianCHK23,DBLP:journals/tcs/JansenN19,DBLP:journals/jgaa/GeffenJKM20} & \cite{DBLP:journals/corr/abs-1103-0534/CyganNPPRW11,DBLP:journals/talg/CyganNPPRW22,DBLP:journals/talg/LokshtanovMS18} & \cite{DBLP:conf/wg/HegerfeldK23,DBLP:journals/siamdm/Lampis20} & \cite{DBLP:conf/esa/HegerfeldK23,DBLP:journals/siamdm/Lampis20}%
\end{tabular}%
\caption{\label{table:tight-bounds}Tight complexity bounds (modulo SETH) for a selection of (connectivity) problems relative to cutwidth, treewidth, modular-treewidth, and clique-width. (Table adapted from~\cite{DBLP:conf/esa/HegerfeldK23}.)}
\end{table}

\paragraph{Our work.}
We close the gap for \textsc{Steiner Tree[$\cw$]} by developing a one-sided Monte Carlo algorithm (with false negatives only) that runs in time $3^{\cw}\cdot n^{\Oh(1)}$. Because $\cw(G)\leq\pw(G)+2$ for each graph $G$ (folklore), the known lower bound for \textsc{Steiner Tree[$\pw$]}~\cite{DBLP:journals/corr/abs-1103-0534/CyganNPPRW11} immediately rules out time $(3-\varepsilon)^{\cw}$ for all $\varepsilon>0$ assuming SETH. In fact, this lower bound holds already relative to cutwidth~\cite{DBLP:conf/stacs/BojikianCHK23}, so we get the same tight complexity of time $3^{k}\cdot n^{\Oh(1)}$ relative to all parameters between cutwidth and clique-width. This is a surprising behavior when compared with the behavior of other connectivity problems (see \cref{table:tight-bounds}) where complexity increases significantly for more general input structure. The following theorem formally states our main result.

\begin{theorem}\label{theorem:intro:mainresult}
 There is a one-sided error Monte Carlo algorithm (no false positives) that, given a graph $G=(V,E)$, a set of terminals $T\subseteq V$, a number $\budget$, and a $k$-clique-expression of $G$, takes time $3^k\cdot n^{\Oh(1)}$ and determines, with high probability, whether a connected subgraph $H$ of $G$ of exactly $\budget$ vertices exists that spans all of $T$.
\end{theorem}

Apart from settling the complexity of \textsc{Steiner Tree[$\cw$]}, the main interest of course lies in how the gap between $\bin$-rank and largest triangular submatrix could be overcome. We should remark that the description of the algorithm does not explicitly talk about any matrices and works entirely on the level of \emph{(connectivity) patterns}, i.e., on how partial solutions connect the different label classes in a partial graph corresponding to a part of the $k$-clique-expression.\footnote{We will still delve a little into the matrix perspective in \cref{sec:technicalcontribution} because we find it more instructive. After all, the project started by trying (and failing) to computationally find larger than $3^k\times 3^k$ triangular submatrices.} We identify a family of patterns, called \emph{complete patterns}, that are representative for the class of all patterns in a strong and constructive sense: For each pattern $p$ there is a set $R_p=\{q_1,\ldots,q_{\ell}\}$ of complete patterns that together completes into exactly the same solutions as $p$ does.
Crucially, we also identify a specific basis $\CSP$ of size $3^k$ to the submatrix induced by complete patterns only; so it has rank~$3^k$.

Hence, one might feel compelled to apply known techniques to leverage this small basis in a tight algorithm. However, a direct application might fall short of achieving this, since using the resulting basis, one can only count representations of partial solutions, but not the solutions themselves (see \cref{sec:technicalcontribution}). Instead, in addition to the usual step of isolating some optimal solution, we set up a second layer of isolation that assigns small random weights to actions in the DP, i.e., to the different contributions in the DP recurrences. The crux is that adding, e.g., a join between labels $i$ and $j$ creates connectivity patterns that are no longer complete and, in general, it takes several complete patterns to represent them. By giving each possible action at a node in the expression tree (there are never more than four different ones) its own random weight, we effectively isolate a single representation of the previously isolated solution (with sufficient probability). Beyond these technical contributions, we apply established tools for fast DP, e.g., fast convolutions~\cite{DBLP:conf/esa/RooijBR09,BjorklundHKKNP16, DBLP:conf/csr/Rooij21, DBLP:conf/esa/HegerfeldK23}. Let us nevertheless remark that our low-rank transformation is not via a corresponding cut-and-count basis but using a different set of (connectivity) states that is related to states in lower bounds in previous work~\cite{DBLP:conf/stacs/BojikianCHK23}.\footnote{This is purely for convenience though it seems interesting that exact cut-and-count is not mandatory. There is no reason to believe that this part of the algorithm could not be done via a cut-and-count-like factorization/basis.}

We think that our work will be instrumental for settling the (parameterized) complexity of further problems relative to less restrictive parameters like clique-width: First, \textsc{Steiner Tree[$\cw$]} isolates the hardness of connectivity and our approach via complete patterns may transfer more or less directly to other connectivity problems parameterized by clique-width. Second, more generally, we think that our method of isolating different representations of a solution via weighted actions in DP is a general and robust way of dealing with further cases, not just for connectivity, where we find a gap between the rank and the largest triangular/permutation submatrices, i.e., when the rank of the matrix is not the correct answer for the (parameter dependence in the) complexity.

\paragraph{Further related work.}
Clique-width was introduced by Courcelle and Olariu~\cite{CourcelleO00} building on work of Courcelle et al.~\cite{DBLP:journals/jcss/CourcelleER93} and it is similar to the NLC-width of Wanke~\cite{DBLP:journals/dam/Wanke94}. Courcelle et al.~\cite{DBLP:journals/mst/CourcelleMR00} showed that every graph problem expressible in MSO$_1$ logic (monadic second order logic of graphs with quantification over vertex sets but not edge sets) can be solved in linear time for graphs with a given $k$-clique-expression. In other words, all these problems are FPT when parameterized by clique-width, time $f(k)\cdot n$, though the function $f$ depends on the formula capturing the problem and may be non-elementary (cf.~\cite{DBLP:conf/icalp/Lampis13}), so likely far from being tight.\footnote{This is not an artifact of clique-width but fully analogous to Courcelle's theorem for graphs of bounded treewidth~\cite{DBLP:journals/iandc/Courcelle90,DBLP:journals/ita/Courcelle92} and corresponding lower bounds~\cite{DBLP:journals/apal/FrickG04}.} This tractability is not restricted to MSO$_1$-expressible problems but many other important problems are \classFPT with respect to clique-width or at least in \classXP, i.e., admitting a time $n^{f(\cw)}$ algorithm (see, e.g.,~\cite{DBLP:conf/wg/EspelageGW01}). Nevertheless, some problems, like \textsc{Disjoint Paths}, are \classNP-complete on graphs of bounded clique-width, while being \classFPT with respect to treewidth (cf.~\cite{DBLP:journals/tcs/GurskiW06}). The first single-exponential time algorithms for connectivity problems parameterized by clique-width were given by Bergougnoux and Kant\'e~\cite{DBLP:journals/tcs/BergougnouxK19,DBLP:journals/siamdm/BergougnouxK21}, e.g., \textsc{Steiner Tree}, \textsc{Connected Dominating Set}, and \textsc{Connected Vertex Cover} each in time $2^{\Oh(\cw)}\cdot n$, but building on the rank-based approach these bounds are likely not tight.

The main drawback of clique-width lies in the difficulty of finding good expressions, i.e., good bounds on the clique-width of given graphs, in reasonable time. Fellows et al.~\cite{DBLP:journals/siamdm/FellowsRRS09} showed that it is \classNP-complete to determine, on input of graph $G$ and integer $k$, whether the clique-width of $G$ is at most $k$. It is open, however, whether for each fixed value of $k$ there is an efficient algorithm for recognizing graphs of clique-width at most $k$; such algorithms are known only for $k\leq 3$~\cite{DBLP:journals/dam/CorneilHLRR12}. In particular, there is neither an \classFPT-algorithm for \textsc{Clique-Width[$k$]} known, nor is it known to be \classW{1}-hard with respect to $k$. The arguably best way for using low clique-width is an exponential-ratio \classFPT-algorithm due to Seymour and Oum~\cite{DBLP:journals/jct/OumS06} (made faster by Oum~\cite{DBLP:journals/talg/Oum08}). That being said, better algorithms for computing clique-width are still possible, and there may be variants of clique-width that give similar complexity (used as parameters) but are easier to compute.

Regarding further parameters relative to which there are tight bounds (modulo SETH) there is another width parameter called \emph{modular-treewidth/pathwidth}~\cite{DBLP:journals/siamdm/Lampis20,DBLP:conf/wg/HegerfeldK23}, which essentially falls between treewidth/pathwidth and clique-width.\footnote{It should be noted that the clique-width of a graph is at most exponential in its treewidth and, unfortunately, this bound is tight, making some comparisons of parameters and bounds a little awkward.} Generally, tight bounds modulo SETH seem mostly confined to parameterization by width parameters, where building gadgets for lower bounds seems easier. There are, though, a few examples where a tight bound holds for a so-called \emph{modulator} parameter or even for parameterization by \emph{solution size} (see, e.g.,~\cite{DBLP:conf/swat/Cygan12,DBLP:conf/stacs/PiecykR21,DBLP:journals/dam/JaffkeJ23,DBLP:conf/iwpec/HegerfeldK22}).

\paragraph{Organization.}
In \cref{sec:technicalcontribution} we provide a short review of our technical contribution, and outline our methods more formally. \cref{sec:pre} introduces the required notation.
In \cref{sec:pats} we define \emph{patterns} and prove a handful of their properties.
In \cref{sec:sol-pat}, we show how to represent partial solutions using patterns, and show that one can build all patterns corresponding to partial solutions over a syntax tree of a clique-expression recursively.
In \cref{sec:rep}, we show how to represent any set of patterns using complete patterns only.
In \cref{sec:rep-count} we show that one can compute the parity of weighted representations of weighted partial solutions using complete patterns recursively over the syntax tree. We also show that using the isolation lemma, one can isolate a single representation of the minimum weight solution with high probability.
In \cref{sec:parity-rep} we show that one can reduce any set of complete patterns into a set of $CS$-patterns in a way that preserves counting (over $\bin$). 
In \cref{sec:algo} we show how to compute the parity of representations of partial solutions using $CS$-patterns efficiently, proving the main theorem of this work.
We conclude in \cref{section:conclusion}.

\section{Technical contribution}\label{sec:technicalcontribution}

In this section we review the technical contribution of this paper. The goal of the section is twofold. The first is to explain the motivation and the intuition behind the presented methods in a concise way. The second is to show the significance of the methods by explaining, on a high level, why existing techniques seem unable to close the gap for \textsc{Steiner Tree[$\cw$]}.

\subparagraph{Patterns.}
Not surprisingly, connectivity patterns play an important role for our work. They get more complex in the context of clique-width (cf.~\cite{DBLP:conf/esa/HegerfeldK23}) and identifying the right families of patterns is crucial. For a natural number $k$, let $[k]_0 = \{0,1,\dots, k\}$. We define a \emph{pattern} as a subset of the power-set of $[k]_0$ with exactly one set containing the element $0$, called the zero-set (\cref{def:pattern}). Let $\Pat$ be the family of all patterns. A labeled graph is an undirected graph $G=(V,E)$ together with a function $\lab\colon V\rightarrow [k]$. Let $\term\subseteq V$ be a non-empty set of terminals (see \cref{sec:pre} for reference). 
Let us fix an arbitrary vertex $v_0 \in \term$.
Then each set $S\subseteq V$ defines a pattern $\pat^S_G$ in $G$, given by the labels appearing in each connected component of $G[S]$. We add the label $0$ to the set corresponding to the connected component of $v_0$, if $v_0\in S$ or as a singleton otherwise (\cref{def:sol-pat}). It holds that $S$ is a Steiner tree in $G$ if and only if $\term \subseteq S$ and $p^S_G$ contains a single set (\cref{lem:solution-if-one-set}). Given a $k$-expression of the graph $G$, and a corresponding syntax tree $\syntaxtree$ (see \cref{sec:pre}). For each $x\in V(\syntaxtree)$ and $\budget\in[|V|]_0$ we define the family (\cref{def:sol})
\[\sol_x[\budget] = \{p^S_{G_x} \colon S\subseteq V_x, |S| = \budget, \term\cap V_x\subseteq S\}.\]
Notably, the patterns $p^S_{G_x}$ carry enough information, so that we could compute $\sol_x[\budget]$, from the families $\sol_{x'}[b']$ for all nodes $x'$ children of $x$ in $\syntaxtree$, and all values of $b'$ only (\cref{lem:recursive-sol}).
However, since the number of patterns is approximately $2^{2^k}$, and the number of subsets $S$ of $V$ is $2^n$, these families are too large to compute. Hence, we follow previous techniques to approach this problem from an algebraic perspective.
For this we require a natural \emph{join} operation for patterns that is suitable for clique-expressions (\cref{def:patops}). Intuitively, the result can be obtained by taking two labeled graphs with partial solutions matching the two patterns and adding all edges between vertices of the same label (i.e., connect vertices of label $i$ of the first graph with vertices of label $i$ in the second); the pattern of the combined (partial) solution is the result of the join. We say that two patterns $p,q$ are consistent ($p\sim q$), if it holds that $p\join q$ contains one set only.

\subparagraph{Algebraic perspective.}
As mentioned before, most algorithmic approaches in this area use some algebraic techniques on certain (consistency) matrices, defined by the given problem and the parameter, to prove some upper-bound on the rank of the matrix, possibly after manipulating the underlying graph \cite{DBLP:conf/stacs/BojikianCHK23} or the matrix itself \cite{DBLP:journals/tcs/JansenN19}, and use the resulting bound to achieve a running time that matches the given rank. 
On the other hand, matching lower bounds mostly use the fact that for the given matrix the size of a maximum triangular (or even permutation) submatrix matches the rank of the matrix itself. This allows to use partial solutions corresponding to the submatrix to encode assignments to variables in a reduction from \textsc{$q$-CNF SAT}, while the triangular/permutation form ensures that these can be propagated through the gadgets in a sufficiently controlled way.

We define the binary matrix $\mat \subseteq \{0,1\}^{\Pat\times \Pat}$, whose rows and columns are indexed by the family of all patterns, where for $p,q\in \Pat$, we define
\[\mat[p,q]=[p\sim q].\]
Intuitively, the matrix indicates for each pair of patterns whether they are consistent or not. This encodes for all labeled graphs, and all possible solutions, whether the union of two solutions induces a connected subgraph. Similar to previous work in this area, we aim to leverage the algebraic properties of this matrix to reduce the space of solutions in a dynamic programming scheme over $\syntaxtree$. However, as we shall see next, the rank of this matrix over $\bin$ is at least $4^k$, while a maximum triangular submatrix has size $3^k$. 
This is the aforementioned gap between the rank and the size of triangular submatrices, which obstructs the application of known techniques.

We define a set of states $\states=\{\stateO, \stateS, \stateC, \stateCS\}$. Let $s\colon [k]\rightarrow \states$ be some state-assignment (we write such assignments as $k$-tuples in the natural way). We define the pattern
\[
p_s = \Big\{ \big\{0\big\}\cup\big\{i\colon s(i)\in\{\stateC, \stateCS\}\big\}\Big\}
\cup \Big\{ \big\{i\big\} \colon s(i)\in \{\stateS, \stateCS\}\Big\},
\]
which intuitively corresponds to a solution consisting of a single (large) connected component containing $v_0$, and some additional isolated vertices. The mapping $s$ then specifies which labels are touched by the connected component containing $v_0$ (and hence, appear in the zero-set), indicated by the state $\stateC$, and which labels appear in components of size one, indicated by $\stateS$, that is confined to isolated vertices with that label. The state $\stateCS$ simply corresponds to both cases at the same time;
We emphasize that $p_s$ consists only of a zero-set and singletons, where labels with $\stateS$ in their state appear as singletons, and labels with $\stateC$ in their state belong to the zero-set of $p_s$.
Let $\Pat_B = \{p_s \colon s\in \states^{[k]}\}$ be the family of patterns corresponding to all state-assignments. It holds that $|\Pat_B| = 4^k$. We show that the submatrix $\mat_B$ of $\mat$, whose rows and columns are indexed by $\Pat_B$ has full rank over $\bin$, which proves that the rank of $\mat$ is at least $4^k$. For a single label, we get the matrix $\mat'$ (See \cref{fig:consistency-matrix}).
\begin{figure}[H]
\[
\hfill
\mat' = 
\begin{bNiceMatrix}[first-row, first-col]
& (\stateC) & (\stateCS) & (\stateO) & (\stateS)\\
(\stateS)  & 1 & 1 & 0 & 0\\
(\stateO)  & 1 & 0 & 1 & 0\\
(\stateCS) & 1 & 1 & 0 & 1\\
(\stateC)  & 1 & 1 & 1 & 1\\
\end{bNiceMatrix}
\hfill
\]
\caption{ 
    \label{fig:consistency-matrix}
    The consistency matrix $\mat'$ indexed by $\Pat_B$ patterns over a single label. We indicate by the given indices the patterns $p_s$ defined by state assignments $s$ as tuples of arity one.}
\end{figure}
We call two states $X,Y\in\states$ \emph{consistent}, if it holds that $\mat'[p_{s_X},p_{s_Y}] = 1$, where $s_Z$ is the mapping that assigns the state $Z\in \states$ to the single label.
It is not hard to see that $\mat'$ has full rank over $\bin$. Now for two arbitrary state-assignments $s,t \in \states^{[k]}$, it holds that $p_s$ and $p_t$ are consistent if and only if $s(i)$ and $t(i)$ are consistent for each $i\in[k]$. Hence, $\mat_B$ is the $k$th Kronecker power of $\mat'$.
It follows that the rank of $\mat_B$ is indeed $4^k$.

\subparagraph{Triangular submatrix.}
Although this does not prove that rank of the whole matrix $\Pat$ is $4^k$, but only sets a lower bound on it, it already suggests that a direct application of known methods would only yield a running time of $\alpha^k\cdot n^{\Oh(1)}$ for $\alpha\geq 4$. On the other hand, it was open whether $3^k$ is the largest achievable dimension of a triangular submatrix of $\Pat$. Let $\statescs = \{\stateS, \stateO, \stateCS\}$, and $\CSP = \{p_s\colon s \in \statescs^{[k]}\}$, the set of all patterns corresponding to state-assignments that map only to $\statescs$. Then the matrix $\mat_{CS}$, the submatrix of $\mat$ induced by patterns in $\CSP$, is a triangular submatrix of size $3^k$. This follows from the fact that the matrix $\mat'_{CS}$ given by one label is a triangular matrix (see \cref{fig:CS-matrix}), and that $\mat_{CS}$ is the $k$th Kronecker power of $\mat'_{CS}$ by the same argument above. We call $\CSP$ the family of $CS$-patterns (\cref{def:cs-pat}).

\begin{figure}[H]
    \[
        \hfill
        \mat'_{CS} = 
    \begin{bNiceMatrix}[first-row, first-col]
    & (\stateCS) & (\stateO) & (\stateS)\\
    (\stateS)  & 1 & 0 & 0\\
    (\stateO)  & 0 & 1 & 0\\
    (\stateCS) & 1 & 0 & 1\\
    \end{bNiceMatrix}
        \hfill
    \]
    \caption{ 
        \label{fig:CS-matrix}
        The matrix $\mat'_{CS}$ indexed by $CS$-patterns over a single label. We indicate by the given indices the patterns defined by assigning the given state to the single label.}
    \end{figure}

Our project of determining the complexity of \textsc{Steiner Tree[$\cw$]} started by trying to find larger triangular submatrices by computational means. While no larger triangular submatrices were found, this led to reduction rules that allowed the submatrix search to safely forget rows/columns corresponding to certain patterns because a largest triangular submatrix could avoid them. For two patterns $p,q \in \Pat$, we say that $p$ \emph{dominates} $q$ ($p\pdom q$), if for each pattern $r$ it holds that $q \sim r$ implies $p\sim r$. For a pattern $p$, if there exists a set of patterns $\{q_1,\dots q_{\ell}\}$, all dominated by $p$, and such that for each pattern $r$ with $p\sim r$, there exists $i\in[\ell]$ with $q_i\sim r$, then there exists an optimal triangular submatrix of $M$ that excludes the pattern $p$ from its rows.
By inspecting the structure of the largest resulting triangular submatrices after exhaustively deleting such patterns, we were able to identify a family of patterns $\Cp \subseteq \Pat$, called the \emph{complete patterns}, such that for each pattern $p\in \Pat$, there exists a set of complete patterns $R_p\subseteq \Cp$, where $p$ dominates all patterns in $R_p$, and for each pattern $r\in \Pat$ with $p\sim r$, there exists $q\in R_p$ with $q\sim r$ as well (\cref{lem:pattern-complete-rep}). We say that $R_p$ \emph{represents} $p$ (\cref{def:representation}), and call the set $R_p$ a \emph{complete representation} of $p$.

On a more technical level, a pattern $p$ is complete if and only if each label that appears in $p$ appears as a singleton as well (\cref{def:complete-patterns}). In order to see why a complete representation always exists, let $p\in \Pat\setminus \Cp$, and let $i$ be a label that appears in $p$ but not as a singleton. Let $r$ be some pattern with $p\sim r$. Let $p_0$ be the pattern resulting from $p$ by removing $i$ from all sets of $p$, and $p_1 = p \cup \big\{\{i\}\big\}$ be the pattern resulting from $p$ by adding the singleton $\{i\}$. We say that $p_0$ results from $p$ by \emph{forgetting} the label $i$, and $p_1$ results from $p$ by \emph{fixing} the label $i$ (\cref{def:forget-fix}). If the label $i$ appears in $r$, then it holds that $p_1\sim r$, while if $i$ does not appear in $r$, then it holds that $p_0\sim r$ (\cref{lem:forget-and-fix-rep}). On the other hand, it holds that both $p_0$ and $p_1$ are dominated by $p$, since by removing an element from a set of $p$ (\cref{lem:removing-element-dom}), or by adding a subset of an existing set to $p$ (\cref{lem:adding-subset-dom}), we only get patterns that are dominated by $p$. Hence, by independently forgetting or fixing each label that appears as non-singleton only in $p$, we get a complete representation of $p$.

Since each non-complete pattern admits a complete representation, this already implies that there exists an optimal triangular submatrix, whose rows avoid all non-complete patterns. Hence, there exists an optimal triangular submatrix, whose rows are indexed by complete patterns only. Let $\mat_C$ be the submatrix of $\mat$ whose rows and columns are indexed by the family of complete patterns. Note that $\CSP$ is a subset of $\Cp$. In \cref{lem:pat-parity-rep} we show that $\CSP$ builds a row basis of $\mat_C$ over $\bin$, by presenting an explicit basis representation for each pattern $p\in\Cp$. Since the rank of a matrix is an upper bound of the size of a maximum triangular submatrix of this matrix, this already implies that $3^k$ is the optimal size of a maximum triangular submatrix of $\mat$.

\subparagraph{Why previous techniques likely do not apply.}
This leaves a gap between the best achievable algorithm using known techniques, with time no better than $4^k\cdot n^{\Oh(1)}$, and the best achievable lower bound (modulo SETH), which only excludes algorithms with running time $(3-\varepsilon)^kn^{\Oh(1)}$. We are only aware of two previous works that managed to overcome such a gap between the size of a triangular submatrix, and the rank of the matrix, solving the \textsc{Feedback Vertex Set} problem~\cite{DBLP:conf/stacs/BojikianCHK23}, and the \textsc{Chromatic Number} problem~\cite{DBLP:journals/tcs/JansenN19}, both parameterized by cutwidth. The former result is based on manipulating the underlying graph, namely subdividing each edge twice, restricting the set of states that can be assigned to a subdivision vertex, and additionally bounding the number of non-subdivision vertices in each cut along a linear arrangement of the resulting graph. The latter manipulates the consistency matrix, indexed by the colorings of the vertices on the two sides of a cut using the polynomial $f_G(x_1,\dots,x_n)=\prod_{(u,v)\in E(G)}(x_u-x_v)$. The idea behind this change, as mentioned by the authors of this paper, is to lower the rank of the matrix, preserving its support. The authors then provide a decomposition of the resulting matrix and bound its inner dimension by $2^k$, where $k$ is the cutwidth of the underlying graph. 

First of all, both techniques are problem-specific, and it is not clear how to use similar ideas to approach the \Stp\ problem. In particular, while all existing optimal algorithms for connectivity problems, to the best of our knowledge, deal with consistency matrices defined over $\bin$, it is clearly impossible to manipulate a binary matrix without altering its support. This leaves open, whether changing the underlying field would allow such a manipulation (similar to \cite{DBLP:journals/tcs/JansenN19}). However, we are not aware of any development in this direction so far, even when considering the broader class of \emph{connectivity} problems, and different structural parameters. Second, and more critical, both techniques are heavily reliant on the small edge cuts along a linear arrangement, and it is not clear how to generalize these techniques to clique-width, especially, since as one can see in \cref{table:tight-bounds}, these techniques already do not generalize to pathwidth and treewidth.

All is not lost, however, because the fact that the family $\Cp$ represents the family of all patterns, and that it induces a ``low-rank'' submatrix of $\mat$, hints at a possible solution. For example, one can try to apply a technique similar to \cite{DBLP:journals/iandc/BodlaenderCKN15}, where, in a bottom-up manner over $\syntaxtree$, we represent the family of solutions using complete patterns only $\rep$, and follow that by computing a basis subset $\rep^*$ of the resulting family of size $3^k$ using Gaussian elimination (noting that the basis must be a subset of $\rep$, since a basis given as a subset of $\CSP$ (the basis for all complete patterns) might result in false positives). Still, it is unclear whether we can avoid the $\omega$ factor (the matrix multiplication exponent) in the exponent of the running time, from applying Gaussian elimination.

Another approach that makes use of low $\bin$-rank matrices is based on defining an explicit family of solutions (usually an explicit basis) that, while not necessarily preserving representation, is preserving the parity of the number of solutions. By a direct application of the Isolation Lemma, one can assume with high probability, that there exists a unique solution of the minimum weight, reducing the decision version of the problem to the counting version. This approach covers, among others, different applications of the \cnc\ technique \cite{DBLP:conf/stacs/BojikianCHK23,DBLP:journals/talg/CyganNPPRW22,DBLP:conf/esa/HegerfeldK23,DBLP:conf/wg/HegerfeldK23}.
While a direct application of such a technique would require at least time $4^k\cdot n^{\Oh(1)}$, we can again try to leverage the aforementioned representation properties to design a faster algorithm. The difficulty in applying such an approach lies in the fact that multiple patterns in a complete representation of a single solution can be consistent with a given pattern. For example, if a label $i$ appears as non-singleton only in a pattern $p$, and for a pattern $r\sim p$, if $i$ appears only in the zero-set of both $r$ and $p$, then both patterns resulting from forgetting or fixing $i$ in $p$ are consistent with $r$. Hence, applying a similar approach to the \cnc\ technique over the set of complete patterns, would only allow us to count (modulo two) the number of representations of solutions, which does not necessarily match the number of solutions themselves.

\subparagraph{Isolating a representative solution.} In order to overcome this problem, we first introduce a set of pattern operations (\cref{def:patops,def:patadd}), that for a node $x\in V(\syntaxtree)$, given a pattern corresponding to some partial solution at each child of $x$, correctly compute the pattern corresponding to the union of these solutions in $G_x$ (\cref{lem:recursive-sol}). We show that all these operations preserve representation (\cref{lem:ops-preserve-rep,lem:patadd-preserve-rep}), noting that complete patterns are closed under the operations corresponding to relabel and union nodes. Join nodes are the key obstruction because the additional connectivity, say from edges between label $i$ and label $j$ vertices, removes singletons in $i$ and $j$ if both labels are in the pattern. In this case, we show that the resulting non-complete patterns can be represented by at most four complete patterns, by independently fixing or forgetting either label (\cref{obs:complete-closed-ops}). Hence, we can make use of the Isolation Lemma, not only to isolate a single minimal-weight solution (similar to previous work), but also by defining a new weight function $\actionf:V(\syntaxtree)\times[4]\rightarrow [\D]$ for some fixed value $\D$, we can isolate a single representation of this solution with high probability (\cref{lem:iso-rep}). We achieve this by defining four new pattern operations called \emph{actions}: Given a pattern $p$, that results from applying the operation corresponding to a join node on a complete pattern, these actions output the four different complete patterns in a complete representation of $p$ (\cref{def:repactions}). By assigning different weights to different actions, we can assign a weight to a pattern at the root node, given by the sum of the weights of all the actions taken to build this pattern along $\syntaxtree$ (\cref{def:action-sequence}).

Finally, we show that all mentioned operations preserve counting (\cref{lem:ops-preserve-parity-rep,lem:patadd-preserve-parity-rep,lem:ac-preserve-parity-rep}) and hence, we can apply a dynamic programming scheme, in a bottom-up manner over $\syntaxtree$ to build a basis-representation that counts the number of representations of weight $\actionweight$ of all solutions of weight $\weight$ that are consistent with a given pattern. We show in \cref{cor:bas-if-drep,lem:no-sol-no-drep} and \cref{lem:iso-total-prob} that it suffices to count such representations to get a one-sided error Monte-Carlo algorithm, that solves the decision version of the \Stp\ problem with high probability. While for a relabel node, the tables can be extended in the natural way, for a join node, we show how to compute a complete representation of the resulting patterns, and a basis representation of it efficiently. Finally, for a union node, we make use of fast convolution over lattices~\cite{BjorklundHKKNP16}.

\section{Preliminaries}\label{sec:pre}
In this work we deal with undirected graphs only. Given a natural number $k$, we denote by $[k] = \{1,2,\dots k\}$ the set of natural numbers smaller than or equal to $k$, and $[k]_0 = [k]\cup\{0\}$. A \emph{labeled graph} is a graph $G=(V, E)$ together with a \emph{labeling function} $\lab\colon V\rightarrow \mathbb{N}$. We usually omit the function $\lab$ and assume that it is given implicitly when defining a labeled graph $G$. We say that $G$ is \emph{$k$-labeled}, if it holds that $\lab(v)\leq k$ for all $v\in V$. We define a \emph{clique-expression} $\mu$ as a well-formed expression that consists only of the following operations on labeled graphs:
\begin{itemize}
    \item \emph{Introduce} $i(v)$ for $i\in\mathbb{N}$. This operation constructs a graph containing a single vertex and assigns label $i$ to this vertex.
    \item The \emph{union} operation $G_1 \union G_2$. The constructed graph consists of the disjoint union of the labeled graphs $G_1$ and $G_2$.
    \item The \emph{relabel} operation $\relabel{i}{j}(G)$ for $i,j \in \mathbb{N}$. This operations changes the labels of all vertices in $G$ labeled $i$ to the label $j$.
    \item The \emph{join} operation $\add i j (G)$ for $i,j\in\mathbb{N}, i\neq j$. The constructed graph results from $G$ by adding all edges between the vertices labeled $i$ and the vertices labeled $j$, i.e.
    \[\add{i}{j}(G) = (V, E \cup \{\{u, v\}\colon\lab(u)=i \land \lab(v)=j\}).\]
\end{itemize}
We denote the graph resulting from a clique-expression $\mu$ by $G_{\mu}$, and the constructed labeling function by $\lab_{\mu}$. We associate with a clique-expression $\mu$ a syntax tree $\syntaxtree_{\mu}$ (we omit the symbol $\mu$ when clear from the context) in the natural way, and to each node $x\in V(\syntaxtree)$ the corresponding operation. For $x\in V(\syntaxtree)$, the subtree rooted at $x$ induces a subexpression $\mu_x$. We define $G_x = G_{\mu_x}$, $V_x = V(G_x)$, $E_x=E(G_x)$ and $\lab_x = \lab_{\mu_x}$.
Finally, we say that a clique-expression $\mu$ is a \emph{$k$-expression} if $G_x$ is a $k$-labeled graph for all $x\in V(\syntaxtree)$. We define the clique-width of a graph $G$ (denoted by $\cw(G)$) as the smallest value $k$ such that there exists a $k$-expression $\mu$ with $G_{\mu}$ is isomorphic to $G$.
We can assume without loss of generality, that any given $k$-expression defining a graph $G=(V,E)$ uses at most $O(|V|)$ union operations, and at most $O(|V|k^2)$ unary operations~\cite{DBLP:journals/tcs/BergougnouxK19, CourcelleO00}.

We define the \textsc{Steiner Tree} problem as follows: Given an undirected graph $G=(V,E)$ a terminal set $\term\subseteq V$ and some value $\budget \in \mathbb{N}$, asked is whether there exists a set of vertices $S\subseteq V$ of size at most $\budget$, such that $\term\subseteq S$ and $G[S]$ is connected.
For $k\in\mathbb{N}$, and $f$ some computable function, we denote by $\ostar(f(k))$ the running time $f(k)\operatorname{poly}(n)$, where $n$ is the size of the input.

Given two sets $S,T$, we denote by $S\triangle T$ the symmetric difference of $S$ and $T$, i.e.\
$S\triangle T = (S\setminus T)\cup (T\setminus S)$. It holds by a simple counting argument that $|S\triangle T|\bquiv |S|+|T|$.

Let $\mathbb{F}$ be a field, and let $S$ be some set. For a vector $T\in \mathbb{F}^{S}$, and some value $x\in S$, we denote by $T[x]$ the element of $T$ indexed by $x$. In general, we will only deal with binary vectors, i.e. $\mathbb{F}=\bin$. By $\overline{0}_{S}$ we denote the vector $T$ with $T[x] = 0$ for all $x\in S$. We omit the subscript $S$ when clear from the context. For two sets $S,T$, let $M \in \mathbb{F}^{S\times T}$ be some matrix, whose rows are indexed by $S$ and columns by $T$. We denote by $M[s,t]$ the element of $M$ indexed by $s$ and $t$.

For some integer $k$, given sets $S_1,T_1,\dots,S_k,T_k$, let $M_i \in \mathbb{F}^{S_i\times T_i}$ for each $i\in[k]$. We define the Kronecker product of $M_1, \dots M_k$ as the matrix $M = M_1\otimes \dots \otimes M_k \in \mathbb{F}^{(S_1\times \dots \times S_k) \times (T_1\times \dots \times T_k)}$, where for $s_i\in S_i$ and $t_i\in T_i$ for all $i\in[k]$ it holds that
\[
  M[(s_1,\dots,s_k),(t_1,\dots,t_k)] = \prod\limits_{i\in[k]}M_i[s_i, t_i].
\]
It is well-known that $\rank(M_1\otimes \dots \otimes M_k) = \rank(M_1) \cdot \rank(M_2) \dots \cdot \rank(M_k)$.

Given a finite ground set $U$ (we call its elements labels), the power-set of $U$ (denoted by $\pow{U}$) is the set of all subsets of $U$.
A partition is a family of pairwise disjoint subsets of $U$ that span all its elements, i.e.\ $P = \{S_1,\dots S_k\}$ is a partition of $U$, if the two following conditions hold:
\begin{itemize}
    \item $S_i \cap S_j = \emptyset$ for all $i,j \in [k], i\neq j$,
    \item and $\dot\bigcup_{i\in[k]}S_i = U$.
\end{itemize}

A \emph{lattice} is a partial order $(\mathcal{L}, \preceq)$ over a set $\mathcal{L}$, such that any two elements have a greatest lower bound (called \emph{meet}), and a least upper bound (called \emph{join}). We call two lattices isomorphic, if the underlying orders are isomorphic.

Let $f:U\rightarrow V$ be a mapping between two sets $U$ and $V$, and $x\notin U$ a new element. Let $i\in V$. We define the extension $f|_{x\mapsto i}:U\cup \{x\}\rightarrow V$ of $f$ by
\[
f|_{x\mapsto i}(y)=
\begin{cases}
i&\colon y = x,\\
f(y)&\colon \text{otherwise}.
\end{cases}    
\]
Let $U_1, U_2$ be two disjoint sets, and $V$ be another set. For two mappings $f_1:U_1\rightarrow V$ and $f_2:U_2\rightarrow V$, we define the disjoint union $f_1\dot\cup f_2:U_1\cup U_2\rightarrow V$ as
\[f_1\dot\cup f_2(x)=\begin{cases}
    f_1(u) &\colon u\in U_1,\\
    f_2(u) &\colon \text{otherwise}.
\end{cases}\]

Let $U, V$ be two ground sets and $f:U^k\rightarrow V$ a mapping. Given  sets $S_1,\dots S_k \subseteq U$, we define
\[
    f(S_1, \dots S_k) = \{f(s_1,\dots s_k)\colon (s_1,\ \dots s_k)\in S_1\times S_2 \times \dots S_k\}.
\]
For a set $S\subseteq V$ we also define
\[
  f^{-1}(S)=\{(s_1,\ \dots s_k)\in S_1\times S_2 \times \dots S_k\colon f(s_1,\dots s_k) \in S\}.  
\]
An exception to this notation, is when we explicitly mention that $f:U\rightarrow V$ is a weight function, and $V\subseteq\mathbb{Z}$. In this case, for some set $S\subseteq U$ we define 
\[f(S) = \sum\limits_{u\in S}f(u),\]
where we compute the sum over $\mathbb{Z}$.

Given a logical expression $\rho$, we denote by $[\rho]$ the Iverson bracket of $\rho$, i.e.
\[
    [\rho] = \begin{cases}
        1 &\colon \operatorname{Val}(\rho) = \operatorname{True},\\
        0 &\colon\text{otherwise}.
    \end{cases}
    \]

\section{Patterns}\label{sec:pats}

\subsection{Definition and terminology}

Along this work, let $U$ be a finite totally ordered ground set. We assume that $U$ is totally ordered, and fix an arbitrary such order over $U$ if none is provided.
In this section we define the main structures we build this paper on, we call them patterns. A pattern, as we shall see later, represents the state of a partial solution, and carries enough information to extend it to a solution over the whole graph.

\begin{definition}\label{def:pattern}
    Let $0$ be an element not in $U$. A \emph{pattern} $p$ is a subset of the power-set of $U\cup\{0\}$, such that there exists exactly one set of $p$ containing the element $0$. With $\Pat(U)$ we denote the family of all patterns over $U$. We omit $U$ when clear from the context. 
    By $Z_p \in p$ we denote the only set of $p$ containing the element $0$, and we call it the \emph{zero-set} of $p$.
    
    Given a pattern $p\in \Pat$, we define $\lbs(p) \subseteq U$ as the set of all labels in $U$ occurring in $p$, and with $\sing(p) \subseteq U$ the set of all labels in $U$ appearing as singletons in $p$, i.e.\
    \begin{itemize}
        \item[] $\lbs(p) = \bigcup\limits_{S\in p} S\setminus \{0\}$,
        \item[]  $\sing(p) = \big\{u\colon \{u\}\in p, u\neq 0\big\}$.
    \end{itemize}
    Note that we define both sets over $U$, and hence exclude the element $0$ from both of them.
\end{definition}

\begin{definition}
    We define a more concise way to write patterns, that we use quite often along this paper. For a pattern $p=\{\{u^1_1, \dots, u^1_{r_1}\}, \dots, \{u^{\ell}_1,\dots, u^{\ell}_{r_{\ell}}\}\}$, we write
    \[p= [u^1_1 u^1_2\dots u^1_{r_1}, \dots, u^{\ell}_1\dots u^{\ell}_{r_{\ell}}],\]
    where we use square brackets to enclose the pattern, we do not use any separator between the elements of the same set, and separate different sets with comas. We also sometimes omit the symbol $0$ when it appears as a singleton.
    
    When we use this notation, we assume that each element is represented by a single symbol only, and hence, there is a unique way of interpreting such a pattern. For example, both $[12]$ and $[0,12]$ denote the pattern $\{\{0\}, \{1, 2\}\}$, while $[01,23]$ denotes the pattern $\{\{0,1\}, \{2,3\}\}$. On the other hand, the pattern $\{\{0, 10\}\}$ cannot be written in a concise way since the element $10$ consists of more than one symbol.
\end{definition}

Next, we introduce pattern operations. These operations will allow us to extend families of patterns over the syntax tree of a clique-expression to build larger solutions recursively. We define these operations in a similar way to the operations on partitions defined in~\cite{DBLP:journals/iandc/BodlaenderCKN15}.

\begin{definition}\label{def:patops}
    Let $U$ be a finite ground sets, and $p,q \in \Pat(U)$. Let $r \in \Pat(U)$ be the pattern resulting from each of the following operations. We define
\begin{center}
    \begin{tabular}{ l p{.82\linewidth} }
        Join:&$r = p\join q$. Let $\sim_I$ be the relation defined over $p\cup q$ where for $S \in p$ it holds that $S\sim_I T$ if $T \in q$ and $S\cap T \neq \emptyset$. We define the equivalence relation $\reachable[p,q]$ (omitting $p$ and $q$ when clear from the context) as the reflexive, transitive and symmetric closure of $\intersects$. Let $\mathcal{R}$ be the set of equivalence classes of $\reachable$, then we define $r = \{\bigcup_{S \in P} S\colon P\in \mathcal{R}\}$, as the unions of the sets in each equivalence class of $\reachable$.\\
        Relabel:&$r = p_{i\rightarrow j}$, for $i,j\in U$, where $r$ results from $p$ by replacing $i$ with $j$ in each set of $p$ that contains $i$.\\
        Union:&$r=p\oplus q$, where $r = (p\setminus \{Z_p\}) \cup (q \setminus \{Z_q\}) \cup \{Z_p\cup Z_q\}$.\\
    \end{tabular}
\end{center} 
\end{definition}

\begin{definition}\label{def:patadd}
    For $i,j\in U, i\neq j$, we define the operation $\patadd_{i,j} p$ for $p \in \Pat$ as 
    \[
    \patadd_{i,j}p =
    \begin{cases}
        p &\colon \{i,j\}\not\subseteq \lbs(p),\\
        p\join [ij] &\colon \text{otherwise}.
    \end{cases}
    \]
    This operation combines all sets containing the labels $i$ or $j$ into one set, if both labels appear in $p$.
\end{definition}

In the following, we provide a new characterization of the operation join, that makes it easier to prove technical results on patterns.

\begin{definition}
Given two patterns $p, q \in \Pat$, we define a \emph{$(p,q)$-alternating walk} as a sequence of sets $S_1, \dots S_n$, such that $S_i \cap S_{i+1} \neq \emptyset$ for all $i\in[n-1]$ and either $S_i \in p$ and $S_{i-1}\in q$ for all even values of $i$ or $S_i \in q$ and $S_{i-1}\in p$ for all even values of $i$. We define the length of a walk as the number of sets on this walk. We call two sets $S,T \in p\cup q$ to be $(p,q)$-connected, if there exists a $(p,q)$-alternating walk from $S$ to $T$.
\end{definition}

\begin{observation}
    It is not hard to see that $(p,q)$-connectedness defines an equivalence relation on the sets in $p\cup q$. Actually, this equivalence relation is exactly the relation $\reachable$ we get from the join operation. Hence, we can characterize the join operation by the unions of sets $S,T \in p\cup q$ such that there exists a $(p,q)$-alternating walk from $S$ to $T$.
\end{observation}
\begin{lemma}\label{lem:join-is-group}
    The join operation over $\Pat$ is commutative and associative.
\end{lemma}
\begin{proof}
    Commutativity follows from the fact that $\reachable$ is an equivalence relation. For associativity, we show that $p\join(q\join r) = (p\join q)\join r$.

    For $S \in p, T \in q \cup r$, let $T'$ be the union of the sets in the equivalence class of $\reachable[q,r]$ containing $T$. We claim that there exists a $(p,q\join r)$-alternating walk from $S$ to $T'$ if and only if there exists a walk from $S$ to $T$ such that for any two consecutive sets $X, Y$ on the walk it holds that there exists $x,y \in \{p,q,r\}, x\neq y, X\in x, Y \in y$. We call such a walk a $(p,q,r)$-alternating walk.

    In order to see this, let $W$ be a $(p,q\join r)$-alternating walk from $S$ to $T'$. For each segment $X,Y,Z$ in $W$ with $Y \in (q\join r)$, we can replace $W$ with a $(q,r)$-alternating walk between two sets $Y_0, Y_k$ such that $X\cap Y_0 \neq \emptyset$ and $Z \cap Y_k \neq \emptyset$. Finally, if one an endpoint of the walk belongs to $r\join q$, let $T'$ be this endpoint, and let $X\in p$ be the set preceding $T'$ if $T'$ is the last set on $W$, or the set following $T'$ if $T'$ is the first set on $W$. Then we can replace $T'$ with a segment $T_0, \dots T_k = T$ that is a $(q,r)$-alternating walk between a set $T_0$ where $T_0\cap X \neq \emptyset$ and any set $T$ in the equivalence class of $\reachable[q,r]$ that unifies to $T'$.

    For the other direction, let $W$ be a $(p,q,r)$-alternating walk from $S$ to $T$. Then divide the walk into segments (possibly of length one) of $(q,r)$-alternating walks separated by single sets of $p$. All segments of each such walk belong to the same equivalence class of $\reachable[q,r]$. By Replacing each such segment with the set resulting from the union of the sets in the equivalence class corresponding to this segment we get a $(p,q\join r)$-alternating walk from $S$ to $T'$.

    Hence, the union of the sets of an equivalence class of $\reachable[p, q\join r]$ is the union of all sets $S, T \in p \cup q \cup r$ such that there exists a $(p,q,r)$-alternating walk from $S$ to $T$, which is the union of the sets of an equivalence class of $\reachable[p\join q, r]$.
\end{proof}
\begin{definition}\label{def:consistency}
Let $p,q\in \Pat$. We say that $p$ and $q$ are \emph{consistent} (denoted by $p\sim q$), if for $r = p\join q$ it holds that $r = \{Z_r\}$ contains the zero-set only, i.e.\ the join operation merges all sets into one set only.
We define the consistency matrix $\mat(U) \subseteq \{0 ,1\}^{\Pat\times \Pat}$, (omitting $U$ when clear from the context), whose rows and columns are indexed by all patterns over $U$, such that $\mat[p, q] = 1$ if and only if $p\sim q$. We have already mentioned this matrix in the introduction, whose rank and its maximum triangular submatrix inspire the results of this paper.
\end{definition}
\begin{definition}\label{def:domination}
    A pattern $p$ \emph{dominates} a pattern $p'$ (denoted by $p \pdom p'$), if for each pattern $q \in \Pat$ it holds $p' \sim q$ implies $p \sim q$. A pattern $p$ strictly dominates $q$ (denoted by $p\sdom q$), if $p \pdom q$ but $q \not \pdom p$. We call $p$ and $p'$ equivalent (denoted by $p\pequiv q$), if both $p\pdom q$ and $q\pdom p$ hold. Clearly, domination is a transitive relation. It is also not hard to see that strict domination is a strict partial ordering, while $\pequiv$ is an equivalence relation. Moreover, domination builds a partial ordering over the equivalence classes of $\pequiv$.
\end{definition}
\subsection{Characteristics of patterns}
In this section, we prove some results on the structure of patterns and the relations between them, in a way that will allow us to manipulate patterns and prove domination and consistency results that will become useful later in this work.
\begin{lemma}\label{lem:patadd-dominates-orig}
    For $i,j\in U, i\neq j$ and $p\in \Pat$, it holds that $\patadd_{i,j}p \geq p$.
\end{lemma}
\begin{proof}
    If $\{i, j\}\not \subseteq \lbs(p)$, then it holds that $\patadd_{i,j}p = p \geq p$. Otherwise, it holds that $\patadd_{i,j}p = p\join[i,j]$. Let $q\in \Pat$ with $p\sim q$. Let $S\in p\join[ij]$ be the set that all sets of $p$ containing $i$ or $j$ combine into. Let $X\in p\join[ij]$ and $Y\in q$. We show that there exists a $(p\join[ij],q)$-alternating walk between $X$ and $Y$. For this sake, let $X' = X$ if $X\neq S$ or $X'$ be some arbitrary set of $p$ containing $i$ or $j$ otherwise. Then there exists a $p,q$ alternating walk between $X'$ and $Y'$ we can turn this walk into a $(p\join[ij],q)$-alternating walk between $X$ and $Y$ be replacing each set of $p$ containing $i$ or $j$ by $S$ over this path.
\end{proof}

\begin{lemma}\label{lem:patadd-is-union-single-otherside}
Given $i,j \in U, i\neq j$, and $p,q\in \Pat$, with $\{i,j\}\subseteq \lbs(p)$. Then $p\join[ij]\sim q$ if and only if $p\sim q\punion[ij,i,j]$
\end{lemma}
\begin{proof}
Assume that $p\join[ij]\sim q$.
First we show that $p\sim q\punion[0,ij]$. Since $\{i,j\}\cap \lbs(p)\neq \emptyset$, there exists a $(p,q\punion[0,ij])$-alternating walk between $\{i,j\}$ and some set of $p$. Now we show that there exist a $(p,q\punion[ij])$-alternating walk between each set of $p$ and each set of $q$. Let $S$ be the set containing $i$ and $j$ in $p\join[i,j]$. Let $X\in p$ and $Y\in q$. Let $X' = S$ if $\{i,j\}\cap S\neq \emptyset$, and $X' = X$ otherwise. Then there exists a $(p\join[ij],q)$-alternating walk from $X'$ to $Y$. We can turn this walk into a $p,q\punion[ij]$ walk from $X$ to $Y$ by replacing each segment $A,S,B$ on the walk, where $A,B\in q$, by the segment $A,A',\{i,j\},B',B$, where $A',B'\in p$ such that $A'\cap A \neq \emptyset$, and both $A'$ and $B'$ intersect $\{i,j\}$. The sets $A',B'$ exist, since $S$ is the union of all sets containing $i$ or $j$ in $p$, and it intersects both $A$ and $B$. Hence, it holds that $p\sim q\punion[0,ij]$.
Now since $i,j\in \lbs(p)$, it holds that there exists a $(p,q\punion[ij,i,j])$-alternating walk from each of $\{i\}$ and $\{j\}$ to some set of $p$.

Now assume that $p\sim q\punion\{i,j,ij\}$. Let $S$ be the set containing $i$ and $j$ in $p\join[ij]$. Let $X\in p\join [ij]$ and $Y\in q$. Let $X' = X$ if $X\neq S$, or $X'$ some arbitrary set of $p$ that contains $i$ or $j$ otherwise. Then there exists a $(p,q\punion\{i,j,ij\})$-alternating walk from $X'$ to $Y$. We can turn this walk into a $(p\join[ij],q)$-alternating walk by replacing each maximal segment $S_1, X_1, S_2, X_2\dots S_r$ with $S$, where $r\in\mathbb{N}$, $S_i \in p$ such that $\{i,j\}\cap S_i \neq \emptyset$, and $X_i \in \{\{i\},\{j\}, \{ij\}\}$.
\end{proof}

\begin{lemma}\label{lem:join-is-join-other-side}
    Given three patterns $p,q,r \in \Pat$. It holds that $p\join q \sim r$ if and only if $p\sim q \join r$.
\end{lemma}
\begin{proof}
    It holds that $p\join q \sim r$ if and only if the pattern $(p\join q)\join r = p \join (q\join r)$ is a singleton, where the equality holds by the associativity of the join operation \cref{lem:join-is-group}, which is the case if and only if $p\sim q \join r$.
\end{proof}

\begin{lemma}\label{lem:missing-singleton-not-consistent}
    Let $p, q \in \Pat$. If $\sing(p)\setminus \lbs(q)\neq \emptyset$, then $p\not \sim q$.
\end{lemma}
\begin{proof}
    Let $u \in \sing(p)\setminus \lbs(q)$.
    It holds that $\{\{u\}\}$ is an equivalence class of the relation $\reachable$, since $\{u\}$ does not intersect any set of $q$. Hence, $p\join q$ contains a set that is not the zero-set.
\end{proof}

\begin{lemma}\label{lem:removing-element-dom}
    Let $p \in \Pat$, $S\in p$ and $u\in S$ some label. Let $S' = S\setminus\{u\}$. Then for $p' = (p\setminus \{S\})\cup\{S'\}$, it holds that $p \pdom p'$.
\end{lemma}
\begin{proof}
    If $S = \{u\}$, then $p'$ is not consistent with any pattern and the lemma holds trivially. Otherwise, let $q\in\Pat$ with $p'\sim q$. Then $p\sim q$ since any $(p',q)$-alternating walk can be turned into a valid $(p,q)$-alternating walk between the same sets by replacing each occurrence of $S'$ by $S$.
\end{proof}

\begin{lemma}\label{lem:adding-subset-dom}
    Let $p \in \Pat$, and $S\in p$. Let $S' \subseteq S$.
    Then for $p' = p\cup \{S'\}$, it holds that $p\pdom p'$.
\end{lemma}
\begin{proof}
    If $\emptyset \in p$, then $p$ is not consistent with any pattern and the lemma holds trivially. Otherwise, let $q\in\Pat$ with $q \sim p'$. Then $q\sim p$ since any $(p',q)$-alternating walk can be turned into a $(p,q)$-alternating walk by exchanging each occurrence of $S'$ by $S$.
\end{proof}

\begin{lemma}\label{lem:missing-label-not-dom}
    Let $p, p' \in \Pat$. If $\lbs(p)\setminus \lbs(p')\neq \emptyset$, then $p' \not \pdom p$.
\end{lemma}
\begin{proof}
    Fix an arbitrary element $u \in \lbs(p)\setminus \lbs(p')$. Let $R = \{0\}\cup\lbs(p)$, let $q = \{R, \{u\}\}$. Then $p'\not\sim q$ by \cref{lem:missing-singleton-not-consistent}, while $p\sim q$, since $R$ admits a $(p,q)$-alternating walk to each set of $p$, while $\{u\}$ admits a $(p,q)$-alternating walk to some set of $p$ containing $u$.
\end{proof}

\begin{lemma}\label{lem:different-sing-not-equiv}
    Let $p, p' \in \Pat$. If $\sing(p)\neq \sing(q)$ then $p \not \pequiv p'$.
\end{lemma}
\begin{proof}
    Assume that $\sing(p)\setminus \sing(p')\neq \emptyset$ (swap $p$ and $p'$ otherwise).
    Fix an arbitrary $u\in \sing(p)\setminus \sing(p')$. If $u \notin \lbs(p')$, then it holds by \cref{lem:missing-label-not-dom} that $p' \not \pdom p$. Otherwise, assume that $u\in \lbs(p')$. We claim that $p\not \pdom p'$. Choose $R = \{0\} \cup (\lbs(p')\setminus \{u\})$, and let $q = \{R\}$ It holds by \cref{lem:missing-singleton-not-consistent} that $p \not\sim q$. We show that $p' \sim q$ holds. Since $u\notin \sing(p')$, each set containing $u$ contains at least one other label as well. Hence, there exists a $(p',q)$-alternating walk of length $1$ from $R$ to each set of $p'$.
\end{proof}

\begin{corollary}\label{cor:quiv-implies-same-sing-lab}
    Given two patterns $p,q \in \Pat$. If $p\equiv q$ holds, then it must hold that $\lbs(p)=\lbs(q)$ and $\sing(p) = \sing(q)$.
\end{corollary}
\begin{proof}
    This follows directly from \cref{lem:different-sing-not-equiv} and \cref{lem:missing-label-not-dom}, since $p\equiv q$ holds if and only if $p\pdom q$ and $q\pdom p$.
\end{proof}

\begin{lemma}\label{lem:remove-extra-labels}
    Let $p,q \in \Pat$. Let $q'$ be the pattern that results from $q$ by removing all labels in $\lbs(q)\setminus \lbs(p)$
    from each set of $q$. Then $p\sim q$ if and only if $p\sim q'$.
\end{lemma}
\begin{proof}
    Let $S = \lbs(q)\setminus \lbs(p)$.
    Any $(p,q)$-alternating walk can be turned into a $(p,q')$-alternating walk and vice versa, by removing all labels in $S$ from each set on this walk. This holds since no intersection between two consecutive sets along this walk contains an element of $S$.
\end{proof}

\begin{definition}\label{def:op-add-i-to-j}
    Given a patten $p\in \Pat$ and two different labels $i, j \in U$. We define the operation $p' = p_{j \curvearrowleft i}$ over patterns, where $p'$ results from $p$ by removing the label $i$ from all sets, and then adding $i$ to each set that contains the label $j$, i.e.
	\[
		p_{j\curvearrowleft i} = \bigcup\limits_{\substack{S\in p\\j\notin S}} S\setminus \{i\}\cup
					\bigcup\limits_{\substack{S\in p\\j\in S}} S\cup \{i\}.
	\]
\end{definition}
\begin{lemma}\label{lem:rel-is-adding-label-other-side}
	It holds for all $p,q \in \Pat$ and $i,j \in U$ with $i\neq j$ that $p_{i\rightarrow j}\sim q$ if and only if $p\sim q_{j\curvearrowleft i}$.
\end{lemma}
\begin{proof}
	Let $q' = \{S\setminus \{i\}\colon S \in q\}$. Since $i\notin \lbs(p_{i\rightarrow j})$, it holds by \cref{lem:remove-extra-labels} that $p_{i\rightarrow j} \sim q$ if and only if $p_{i\rightarrow j}\sim q'$. We show that this is the case if and only if $p \sim q_{j\curvearrowleft i}$.
    Let $p'=p_{i\rightarrow j}$. For $S \in p$, let $S'$ be the set resulting from $S$ by replacing $i$ with $j$ (if $i\in S$).
    For $T \in q$, let $T' = T$ if $j\notin T$ or $T'=T\cup \{i\}$ otherwise.
    Assuming $p'\sim q'$, we show that $p\sim q_{j\curvearrowleft i}$. Let $S \in p$ and $T' \in q_{j\curvearrowleft i}$. There exists a $(p',q')$-alternating walk from $S'$ to $T$. We turn this walk into a $(p,q_{j \curvearrowleft i})$-alternating walk, from $S$ to $T'$, by replacing each set on odd position with the set it originates from, and by adding the label $i$ to each set with label $j$ at an even position of the walk. Clearly, if two consecutive sets had $j$ in there intersection, they will have $i$ or $j$ in their intersection on the resulting path.
    
    For the other direction, assume that $p \sim q_{j \curvearrowleft i}$. We show that $p'\sim q'$. Let $S' \in p'$ and $T \in q'$. There exists a $(p,q_{j \curvearrowleft i})$-alternating walk from $S$ to $T'$. We can turn this walk into a $(p',q')$-alternating walk form $S'$ to $T$ by replacing each set $S$ on an edd position by $S'$ (turning each $i$ into $j$) and replacing each set $T'$ at an even position with the set $T$ it originated from. Clearly, any two consecutive sets that had $i$ or $j$ in their intersection, have $j$ in their intersection in the resulting walk.
\end{proof}
\begin{lemma}\label{lem:union-to-disjoint}
    Let $p, q\in \Pat(U)$. Let $U'$ be a new label set disjoint from $U$ and $\rho$ a one-to-one mapping from $U$ to $U'$.
    Let $q'\in \Pat(U \cup U')$ the pattern resulting from $q$ by relabeling $u$ to $\rho(u)$ for each $u\in U'$. Then for each pattern $r\in\Pat(U)$ it holds that $p\punion q \sim r$ if and only if $p\punion q' \sim r'$ for $r'\in \Pat(U\cup U')$ the pattern resulting from $r$ by adding $\rho(u)$ to all sets containing $u$ for each $u\in U'$.
\end{lemma}
\begin{proof}
    Let $S \in p\punion q$ and $T\in r$. Let $S', T'$ be the sets corresponding to $S$ and $T$ in $p\punion q'$ and $r'$ respectively. We show that one can turn a $(p\punion q, r)$-alternating walk from $S$ to $T$ to a $(p\punion q', r')$-alternating walk from $S'$ to $T'$ and vice versa. We achieve this by replacing each set in $q$ or $r$ by the corresponding set in $q'$ or $r'$. For two consecutive sets $X, Y \in (p\punion q) \cup r$ on the walk, let $X'$ and $Y'$ be the corresponding sets in $(p\punion q') \cup r'$ respectively. Then $X\cap Y$ contains an element $u$ of $\lbs(p)$, if and only if  $u \in X' \cap Y'$. Otherwise both sets must contain an element $u\in \lbs(q)$. But this is the case if and only if $\rho(u) \in X'\cap Y'$.
\end{proof}
\begin{lemma}\label{lem:disjoint-union-is-join-other-side}
    Given two patterns $p, q\in \Pat$ such that $\lbs(p)\cap \lbs(q) = \emptyset$. It holds for all $r\in \Pat$ that $p\punion q \sim r$ if and only if $p\sim q\join r$.
\end{lemma}
\begin{proof}
    For a set $S$ in $q\cup r$, let $S^*$ the set in $q\join r$ that $S$ unifies to. Assume that $p\punion q \sim r$, then for each set $S$ of $q$ there exists a set $T\in r$ such that $S\cap T \neq \emptyset$. Hence for each set $S^* \in q\join r$ there is a set $S\in r$ such that $S\subseteq S^*$. Let $S\in p$ and $T^* \in p\join r$. Let $T\in r$ a subset of $T^*$. There exists a $(p\punion q, r)$-alternating walk from $S$ to $T$. This walk consists of segments of $(q,r)$-alternating walks separated by single sets of $p$. Note that all sets of each such segment belong to the same equivalence class of $\reachable[q,r]$. Hence, we can replace each such segment by $S^*$, the union of the sets in the corresponding equivalence class, getting a $(p,q\join r)$-alternating walk from $S$ to $T^*$.

    Now assume that $p\sim q\join r$. Since $\lbs(p)\cap \lbs(q) = \emptyset$, it must holds for each set $S\in q$ that there exists a set $T\in r$ with $S\cap T \neq \emptyset$. We show that there exists a $(p\punion q, r)$-alternating walk from each set $S\in p$ to each set $T\in r$. Since each set of $q$ intersects some set of $r$, the claim follows. Let $S\in p$ and $T \in r$. Let $T^* \in q \join r$ such that $T\subseteq T^*$. There exists a $(p,q\join r)$-alternating walk from $S$ to $T^*$. We get a $(p\punion q,r)$-alternating walk from $S$ to $T$ from this walk, by replacing $Z^* \in q\join r$ for each segment $X,Z^*,Y$, by a $(q,r)$-alternating walk from a set $L\in r$ that intersects $X$ to a set $R\in r$ that intersects $Y$. Such a walk exists since $\lbs(p) \cap \lbs(q) = \emptyset$ and $S^*$ intersects both $X$ and $Y$. Also, for $X, T^*$ the last two sets on the walk, we replace $T^*$ with a $(q,r)$-alternating walk from $X' \in r$ to $T$, where $X'\cap X\neq \emptyset$. Such a walk also exists since $X\cap T^* \neq \emptyset$, and $T\subseteq T^*$, and hence, it admits a $(q,r)$-alternating walk to each set that unifies to $T^*$.
\end{proof}
\begin{corollary}\label{cor:union-is-join-other-side}
    Given two patterns $p, q\in \Pat(U)$.
    Let $U'$ be a new label set disjoint from $U$ and $\rho$ a one-to-one mapping from $U$ to $U'$.
    Let $q'$ be the pattern resulting from $q$ by relabeling $u$ to $\rho(u)$ for each $u\in U$. Then it holds for each pattern $r\in\Pat(U)$ that $p\punion q \sim r$ if and only if $p \sim q' \join r'$ for $r'\in \Pat(U\cup U')$ the pattern resulting from $r$ by adding $\rho(u)$ to each set containing $u$ for each $u\in U$.
\end{corollary}
\begin{proof}
    It holds that $p\punion q \sim r$ if and only if $p\punion q' \sim r'$ by \cref{lem:union-to-disjoint}. Since $\lbs(p)\cap \lbs(q')=\emptyset$, it holds by \cref{lem:disjoint-union-is-join-other-side} that $p\punion q' \sim r'$ if and only if $p\sim q'\join r'$.
\end{proof}

\section{Partial solutions as patterns}\label{sec:sol-pat}
From now on, the label set $U$ will correspond to the labels of a labeled graph. Let $(G, \term, \budget)$ be the given instance of the \Stp\ problem, for some graph $G = (V, E)$, a set of terminals $\term \subseteq V$, and $\budget \in \mathbb{N}$. Let $n = |V|$ and $m = |E|$. Let $\mu$ be a $\cwval$-expression of $G$ for some value $\cwval\in\mathbb{N}$. Let $\syntaxtree$ be the corresponding syntax tree, and $r$ its root. 
Hence, we assume hereinafter that $U = [\cwval]$. In particular, $U$ is finite and totally ordered.

We also fix a value $\W\in\mathbb{N}$, and a weight function $\weightf:V\rightarrow[\W]$. Both will be chosen later. Let $v_0\in \term$ be an arbitrary but fixed terminal vertex. For a node $x\in V(\syntaxtree)$, we define $\term_x = \term \cap V_x$.

\begin{definition}\label{def:sol-pat}
    For a node $x\in V(\syntaxtree)$, we call a set $S\subseteq V_x$ where $\term_x \subseteq S$ a \emph{partial solution}. Each partial solution defines a pattern $\pat = \pat_x(S)$ over $G_x$ as follows: Let $C_1, \dots C_{\ell}$ be the connected components of $G_x[S]$. For $i\in[\ell]$, let $S_i = \lab_x(C_i)$ be the set of all labels appearing in $C_i$. If $v_0\notin S$, we define 
    \[\pat = \big\{S_1, \dots, S_{\ell}\big\} \cup \big\{\{0\}\big\}.\]
    Otherwise, assume that $v_0 \in C_{\ell}$ (otherwise swap $C_{\ell}$ and the component containing $v_0$). We define
    \[\pat = \big\{S_1, \dots, S_{\ell-1}\big\} \cup \big\{S_{\ell} \cup \{0\}\big\}.\]
    That means the zero-set is defined by the component containing $v_0$ if $v_0 \in S$, or as a singleton otherwise. We call $\pat_x(S)$ the pattern corresponding to $S$ in $G_x$.    
    We denote by $\pat(S)$ the pattern $\pat_r(S)$ corresponding to a parital solution $S$ in the whole graph $G$.
\end{definition}
\begin{lemma}\label{lem:solution-if-one-set}
    Given a set $S\subseteq V$ such that $\term \subseteq S$, it holds that $S$ is a Steiner tree in $G$, if and only if $\pat(S)$ consists of one set only.
\end{lemma}
\begin{proof}
    Let $p = \pat(S)$. Assume that $S$ is a Steiner tree in $G$. Then it holds that $G[S]$ consists of one connected component only. Since it holds that $v_0\in \term\subseteq S$, it follows that the label $0$ is contained in the same set corresponding to this connected component in $p$, and hence, $p$ consists of one set only.

    Now assume that $p$ consists of one set only, namely the zero-set $Z_p$.
    Even though multiple connected components might correspond to the same set in $p$, there must exist exactly one connected component containing the vertex $v_0$. Hence, at most one connected component in $G[S]$ corresponds to the zero-set in $p$. Since the zero-set is the only set in $p$, it must hold that $G[S]$ contains one connected component only.
\end{proof}
Now we define families of partial solutions that allow us to build a recursive formula for a dynamic programming scheme over $\syntaxtree$ to solve the \Stp\ problem.
\begin{definition}\label{def:sol}
    For $x\in V(\syntaxtree), b \in [\budget]_0, c \in [n\cdot \W]_0$,
    we define the family $\sol_x\ind{b, c}\subseteq \Pat$ of patterns corresponding to partial solutions of cardinality $b$ and weight $c$ over $G_x$ as 
    \[
    \sol_x\ind{b, c} = \{p\in \Pat\colon \exists S\subseteq V_x, \text{ where } \term_x \subseteq S\land p_x(S) = p\land |S|=b \land \weightf(S) = c\}.
    \]
\end{definition}

From now on, we always assume that $x\in V(\syntaxtree), b \in [\budget]_0$, and $c \in [n\cdot \W]_0$. We skip repeating this to avoid redundancy.

\begin{lemma}\label{lem:recursive-sol}
    The families $\sol_x$ can be built recursively over the nodes of $\syntaxtree$ in a bottom-up manner using the operations defined in \cref{def:patops} and \cref{def:patadd}.
\end{lemma}
\begin{proof}
    We distinguish the different types of nodes of $\syntaxtree$: For an introduce node $i(v)$ (a leaf of $\syntaxtree$), let $\weight = \weightf(v)$. Let $p = [0,i]$ if $v \neq v_0$, and $p=[0i]$ otherwise. Then we set $\sol_x\ind{1,\weight} = \{p\}$. We set $\sol_x\ind{0,0}$ to $\{[0]\}$, if $v \notin \term$, and to $\emptyset$ otherwise. For all other values of $\budget$ and $\weight$, we set $\sol_x\ind{\budget,\weight}= \emptyset$.
    
    For a join node $\mu_x = \add{i}{j}(\mu_{x'})$, we define 
    \[
        \sol_x\ind{\budget,\weight}= \patadd_{i,j}\big(\sol_{x'}\ind{\budget,\weight}\big),
    \]
    and for a relabel node $\mu(x) = \relabel{i}{j}(\mu_{x'})$. We define
    \[
        \sol_x\ind{\budget,\weight}= \big(\sol_{x'}\ind{\budget,\weight}\big)_{i\rightarrow j}.
    \]
    For a union node $\mu_{x} = \mu_{x_1}\union\mu_{x_2}$, we define
    \[
        \sol_x\ind{\budget,\weight} = \bigcup\limits_{\substack{b_1+b_2 = \budget\\c_1+c_2=\weight}}\big(\sol_{x_1}\ind{b_1,c_1}\punion\sol_{x_2}\ind{b_2, c_2}\big).
    \]

We prove the correctness of these formulas by induction over $\syntaxtree$. For a leaf node, the correctness follows directly from the definition of $\sol_x\ind{\budget, \weight}$, since the label $0$ is either a singleton, or contained in the set that corresponds to the connected component containing $v_0$.
If $x$ is a join node, it holds that $\lab_x = \lab_{x'}$. Let $S$ be a solution, and $p = \pat_{x'}(S)$. If $S$ does not contain vertices of both labels $i$ and $j$, then no edges are added to $G_{x'}[S]$. Hence, all connected components stay the same. That means $\pat_{x}(S)=p$ whenever $\{i,j\}\not\subseteq \pat()$. Otherwise, it holds that $\{i,j\}\subseteq \lbs(p)$, and all connected components containing either labels unify to one connected components. All other components stay untouched. In this case, it holds that $\pat_{x}(S) = p\join[ij]$. In total, we get $\pat_x(S)=\patadd_{i,j}p$.

In the case of a relabel node, it holds that 
\[
\lab_x(v)= \begin{cases}
    \lab_{x'}(v) &\colon \lab_{x'}(v)\neq i,\\
    j &\colon \text{otherwise}.
\end{cases}    
\]
Hence, $\pat_x(S) = \big(\pat_{x'}(S)\big)_{i\rightarrow j}$.

Finally, for a union node $\mu_{x} = \mu_{x_1}\union\mu_{x_2}$. Each partial solution $S$ over $G_x$ of size $\budget$ and weight $\weight$ results from the union of a partial solution $S_1$ of size $b_1$ and weight $c_1$ in $G_{x_1}$ and a partial solution $S_2$ of size $b_2$ and weight $c_2$ in $G_{x_2}$, where $b_1+b_2=\budget$ and $c_1+c_2=\weight$ and vice versa.

Since at most one of $V_{x_1}$ and $V_{x_2}$ contains $v_0$, at least one of $Z_{p_1}$ or $Z_{p_2}$ is the singleton $\{0\}$. Since the connected components of $G_x[S]$ are the union of both the connected components of $G_{x_1}[S_1]$ and those of $G_{x_2}[S_2]$, it holds that $p_1\punion p_2 = \pat_x(S_1\cup S_2)$. Hence, the formula follows.
\end{proof}

\begin{lemma} \label{lem:solution-if-sol-consist-0}
    The graph $G$ admits a Steiner tree of size $\budget$ and weight $\weight$, if and only if there exists a pattern $p \in \sol_r\ind{\budget, \weight}$ with $p\sim [0]$.
\end{lemma}
\begin{proof}
Since we assume that all patterns contain exactly one set with the label $0$ in it, it holds for all $p\in \Pat$ that $p\join[0] = p$. Hence, $p\sim [0]$ if and only if $p$ contains one set only.
Let $S\subseteq V$ with $\term \subseteq S$, $|S|=\budget$ and $\weightf(S)=\weight$. Let $p=\pat(S)$. It follows from \cref{def:sol} that $p\in \sol\ind{\budget, \weight}$. It follows from \cref{lem:solution-if-one-set} that $S$ is a solution if and only if $p$ is a singleton, which is the case if and only if $p\sim [0]$.
\end{proof}

Hence, we can solve the \Stp\ problem by dynamic programming over $\syntaxtree$ by computing all families $\sol_x\ind{\budget, \weight}$ for all $\budget \in [n]$ and  $\weight \in [n\cdot\W]$. However, since we have approximately $2^{2^k}$ different patterns over $k$ labels, we seek to reduce the families $\sol\ind{\budget, \weight}$ into representing families $\rep\ind{\budget, \weight}$ defined over a special family of patterns. We show that we can count the number of representations of a partial solutions over a node $G_x$ in time $\ostar(3^k)$, which will allow us to solve the \Stp\ problem in this running time with high probability.

\section{Representation of partial solutions}\label{sec:rep}
In this section, we define a special family of patterns (called \emph{complete patterns}). We show that we can represent any family of patterns by a family of complete patterns only, in a way that carries the extendibility of a partial solution into a solution over the whole graph.
\begin{definition}\label{def:representation}
    Given two families $\sol, \rep \subseteq \Pat$, we say that $\rep$ \emph{represents} $\sol$, if for each $q \in \Pat$ the following holds: there exists a pattern $p \in \sol$ such that $p\sim q$ if and only if there exists a pattern $p'\in \rep$ such that $p'\sim q$. Given a family $\rep\subseteq \Pat$, and a pattern $p \in \Pat$, we say that $\rep$ represents $p$ if $\rep$ represents $\{p\}$.
\end{definition}
Clearly, it holds that representation is an equivalence relation. We show that the operations in \cref{def:patops} preserve representation, which will allow us to define a dynamic programming scheme over families of representative solutions.
\begin{definition}
    Let $\op:\Pat^k \rightarrow \Pat$ be a $k$-ary operation over patterns, for some positive integer $k$. We say that $\op$ \emph{preserves representation}, if for all sets of patterns $P_1, \dots P_k, Q_1, \dots Q_k \subseteq \Pat$, where $P_i$ represents $Q_i$ for all $i \in [k]$, it holds that $\op(P_1,\dots P_k)$ represents $\op(Q_1, \dots Q_k)$.
\end{definition}
The following observation follows from the definition of representation.
\begin{observation}
    A $k$-ary operation $\op$ over patterns preserves representation, if and only if for all sets of patterns $P_1, \dots P_k, Q_1, \dots Q_k \in \Pat$, where $P_i$ represents $Q_i$ for all $i \in [k]$, and for all $r \in \Pat$ the following holds: there exists a tuple $(p_1, \dots p_k)\in P_1\times P_2 \times \dots P_k$ such that $\op(p_1, \dots p_k)\sim r$ if and only if there exists a tuple $(q_1,\dots q_k)\in Q_1\times Q_2 \times \dots Q_k$ such that $\op(q_1, \dots q_k)\sim r$.
\end{observation}
\begin{lemma}\label{lem:ops-preserve-rep}
    All operations defined in \cref{def:patops} preserve representation.
\end{lemma}
\begin{proof}
    Let $S,S',T,T'\subseteq \Pat$ such that $S'$ represents $S$ and $T'$ represents $T$.
    
    Join:
    Let $r\in \Pat$ such that there exist $p \in S$ and $q\in T$ where $p\join q\sim r$.
    It holds by \cref{lem:join-is-join-other-side} that, $p\sim q\join r$. Since $S'$ represents $S$, there exists $p' \in S'$ with $p'\sim q\join r$. By the same lemma it follows that $q\sim p'\join r$, and there exists $q' \in T'$ such that $q'\sim p'\join r$. Hence, it holds again by the same lemma that $p'\join q'\sim r$, where $p'\join q'\in S'\join T'$. The other direction follows in an analogous way, by assuming that $p\in S'$ and $q\in T'$.

    Relabel:
    Let $q\in\Pat$ such that there exists $p\in S$, where $p_{i\rightarrow j}\sim q$.
    By \cref{lem:rel-is-adding-label-other-side} it holds that $p\sim q_{j\curvearrowleft i}$. Hence, there exists $p'\in S'$ with $p'\sim q_{j\curvearrowleft i}$.
    Again, by lemma \cref{lem:rel-is-adding-label-other-side} is holds that $p'_{i\rightarrow j}\sim q$, for $p'_{i\rightarrow j} \in S'_{i\rightarrow j}$. The other direction follows analogously.

    Union:
    Let $r \in \Pat$ such that there exists $p\in S, q\in T$ with $p\punion q \sim r$.
    It holds by \cref{cor:union-is-join-other-side} that $p\sim q'\join r'$, where $q'$ and $r'$ are the patterns resulting from $q$ and $r$ as described in the lemma. Hence, there exists $p'\in S'$ such that $p'\sim q\join r$. Note that $p'$ is also a pattern over $U$, and hence, again by \cref{cor:union-is-join-other-side}, it holds that $p'\punion q\sim r$. By an analogous argument (replacing $p$ with $q$ and $q$ with $p'$) we get that $p'\punion q' \sim r$ for some $q' \in T'$. The other direction follows analogously.
\end{proof}
\begin{lemma}\label{lem:patadd-preserve-rep}
    For $i,j\in U, i\neq j$, the operation $\patadd_{i,j}$ preserves representation.
\end{lemma}
\begin{proof}
    Let $S, R \subseteq \Pat$, such that $R$ represents $S$. We show that $\patadd_{i, j}R$ represents $\patadd_{i,j}S$. Let $q\in \Pat$, such that there exists $p\in S$ with $\patadd_{i,j}p\sim q$. If $\{i,j\} \not \subseteq \lbs(p)$, then it holds that $\patadd_{i,j}p = p \sim q$, hence there exists $p'\in R$ with $p'\sim q$. By \cref{lem:patadd-dominates-orig}, it holds that $\patadd_{i,j}p' \sim q$.

    Now assume that $\{i,j\}\subseteq \lbs(p)$. By \cref{lem:patadd-is-union-single-otherside} it holds that $p\sim q\punion[ij,i,j]$. Hence, there exists $p'\in R$ with $p'\sim q\punion[ij,i,j]$. Since $\{i,j\} \subseteq \sing(q\punion[ij,i,j])$, it holds by \cref{lem:missing-singleton-not-consistent} that $\{i,j\}\subseteq \lbs(p')$. It follows from \cref{lem:patadd-is-union-single-otherside} that $\patadd_{i,j}p' \sim q$.
\end{proof}

Now we define a special family of patterns called the complete patterns, and show that we can represent any family of patterns using complete patterns only.
\begin{definition}\label{def:complete-patterns}
    A pattern $p\in \Pat$ is \emph{complete}, if $\lbs(p)=\sing(p)$. We denote by $\Cp(U)$ the family of all complete patterns over $U$.
\end{definition}
This means that a pattern $p$ is complete, if each label of $p$ appears as a singleton in $p$ as well.
\begin{lemma}\label{lem:complete-consistentcy}
    For all $p,q \in \Cp$ the following holds: $p\sim q$ implies that $\lbs(p)=\lbs(q)$.
\end{lemma}
\begin{proof}
    This follows from the fact that $\lbs(p)=\sing(p)$ and $\lbs(q)=\sing(q)$, together with \cref{lem:missing-singleton-not-consistent}.
\end{proof}

\begin{definition}\label{def:forget-fix}
    Let $p\in \Pat$. For $i\in U$, we define the operations $\fix(p,i)$ and $\forget(p,i)$ as
    \begin{alignat*}{2}
        &\fix(p,i) &&= p\cup \big\{\{i\}\big\},\\
        &\forget(p,i) &&=\big\{S\setminus\{i\}\colon S\in p\big\},
    \end{alignat*}
    if $i \in \lbs(p)\setminus \sing(p)$, and $\fix(p,j) = \forget(p, j) = p$ if $j\notin \lbs(p)\setminus\sing(p)$.
    We say that $p' = \fix(p, i)$ results from $p$ by \emph{fixing} the label $i$, and $p'' = \forget(p,i)$ results from $p$ by \emph{forgetting} the label $i$.
    Given a pattern $p\in \Pat$, we define $\inc(p)= \lbs(p)\setminus \sing(p)$.
\end{definition}

\begin{lemma}\label{lem:forget-and-fix-rep}
    Let $p\in \Pat$, and $i \in U$. It holds that $\{\fix(p, i), \forget(p, i)\}$ represents $p$.
\end{lemma}
\begin{proof}
    Let $p' = \fix(p, i)$, and $p'' = \forget(p,i)$.
    It holds by \cref{lem:adding-subset-dom} that $p\pdom p'$, and by \cref{lem:removing-element-dom} that $p\pdom p''$. Now let $q\in \Pat$ with $p\sim q$.
    If $i \notin \lbs(q)$, then it must hold that $p''\sim q$. In order to see this, let $S, T \in p'' \cup q$. Let $S', T'$ be two sets in $p\cup q$ that $S$ and $T$ result from by possibly removing $i$ from the sets. Then there exists a $(p,q)$-alternating walk from $S'$ to $T'$. This walk can be turned into a $(p'',q)$-alternating walk walk between $S$ and $T$ by removing $i$ from all sets on the walk.
    
    Now assume that $i\in \lbs(q)$. We claim that $p'\sim q$. Any $(p,q)$-alternating walk between two sets in $p\cup q$ is a $(p',q)$-alternating walk between the same sets. Hence, there is a walk between any pair of sets other than $\{i\}$. However, since $i\in \lbs(q)$, there exists $S\in q$ with $i\in S$. That means $\{i\},S$ is a $(p',q)$-alternating walk of length two between $\{i\}$ and $S$ and hence, $p'\sim q$ holds.
\end{proof}
\begin{definition}\label{def:complete-rep}
    Given a pattern $p\in \Pat$. Let $\ell = |\inc(p)|$, and $i_1, \dots i_{\ell}$ be the elements of $\inc(p)$ in an increasing order. Let $R_0 = \{p\}$, and $R_{j} = \forget(R_{j-1}, i_j) \cup \fix(R_{j-1}, i_j)$ for all $j\in[\ell]$. We define the family $R_p = R_{\ell}$ and call it a \emph{complete representation} of $p$.
\end{definition}
\begin{lemma}\label{lem:pattern-complete-rep}
    For each pattern $p\in \Pat$, the family $R_p$ has size at most $2^{|\inc(p)|}$, it represents $p$, and contains complete patterns only.
\end{lemma}
\begin{proof}
     It holds by \cref{lem:forget-and-fix-rep} that $R_j$ represents $R_{j-1}$ and hence, by transitivity of representation, $R_j$ represents $p$ for each $j\in[\ell]$. It can be easily seen by induction over $j\in[\ell]_0$, that $\inc(q)\subseteq\{i_{j+1}, \dots i_{\ell}\}$ for all $q \in R_j$. Hence, it holds for $q\in R_{\ell}$ that $\inc(q)=\emptyset$, and hence, $q\in \Cp$. Since $|R_j|\leq 2|R_{j-1}|$ and $|R_0| = 1$, it holds that $R_{\ell} \leq 2^{\ell}$.
\end{proof}
\begin{corollary}\label{cor:family-complete-rep}
    Let $\sol \subseteq \Pat$. Then the family $\rep = \bigcup_{p\in \sol}R_p \subseteq \Cp$ contains complete patterns only, and it represents $\sol$.
\end{corollary}
\begin{proof}
  It holds by \cref{lem:pattern-complete-rep} that $\rep \subseteq \Cp$. Let $q\in \Pat$ such that there exists $p'\in \rep$ with $p'\sim q$. It follows from the definition of $\rep$, that there exists $p\in \sol$ with $p'\in R_p$. Hence, it follows from \cref{lem:pattern-complete-rep} that $p\pdom p'$, and hence, $p\sim q$. On the other hand, let $q\in \Pat$ with $p\sim q$ for some $p\in \sol$. Then there exists by \cref{lem:pattern-complete-rep} a pattern $p'\in R_p\subseteq \rep$ with $p'\sim q$.
\end{proof}

\begin{observation}\label{obs:complete-closed-ops}
    It holds that the family of complete patterns is closed under both relabel operation $i\rightarrow j$ for all $i,j\in U$, and union operation $\punion$. On the other hand, it holds for $i,j\in U, i\neq j$, $R\subseteq \Cp$, $R' = \patadd_{i,j}R$ and for each $p\in R'$ that
    $\inc(p)\in \{\emptyset, \{i,j\}\}$. If $\inc(p) = \{i,j\}$, then $R_p$ contains at most four patterns, that result from $p$ by fixing or forgetting either labels independently.
\end{observation}
Now we show that we can restrict the families of partial solutions $\sol_x\ind{\budget,\weight}$ to representing families $\rep_x\ind{\budget, \weight}$ that contain complete patterns only.
\begin{definition}\label{def:rep-family}
    For $x\in V(\syntaxtree), \budget \in [n], \weight \in [n\cdot \W]$,
    we define the families $\rep_x\ind{\budget, \weight}$ as follows:
        For an introduce node $\mu_x = i(v)$, let $\weight = \weightf(v)$. We set 
        \[
        \rep_x\ind{1,\weight} = 
        \begin{cases}
        \{[0,i]\} & \colon v \neq v_0,\\
        \{[0], [0i,i]\} &\colon \text{otherwise}.
        \end{cases}
        \]
        We set $\rep_x\ind{0,0}$ to $\{[0]\}$, if $v \notin \term$, and to $\emptyset$ otherwise. For all other values of $\budget$ and $\weight$, we set $\rep_x\ind{\budget,\weight}= \emptyset$.
        
    For a relabel node $\mu_x = \relabel{i}{j}(\mu_{x'})$, we set
    \[
        \rep_x\ind{\budget,\weight}= \big(\rep_{x'}\ind{\budget,\weight}\big)_{i\rightarrow j}.
    \]
    For a union node $\mu_x = \mu_{x_1}\union \mu_{x_2}$, we set
    \[
        \rep_x\ind{\budget,\weight} = \bigcup\limits_{\substack{b_1+b_2 = \budget\\c_1+c_2=\weight}}\big(\rep_{x_1}\ind{b_1,c_1}\punion\rep_{x_2}\ind{b_2, c_2}\big).
    \]
    Finally, for a join node $\mu_x = \add{i}{j}(\mu_{x'})$, we define the families
    \[
        \rep'_x\ind{\budget,\weight}= \patadd_{i,j}\big(\rep_{x'}\ind{\budget,\weight}\big),
    \]
    we define the family 
    \[
        \rep_x\ind{\budget,\weight}= \bigcup\limits_{p\in \rep'_x\ind{\budget, \weight}}R_p.
    \]
\end{definition}
\begin{lemma}\label{lem:rep-represent-sol}
    For each $x\in V(\syntaxtree)$ and all values of $\budget, \weight$, it holds that $\rep_x\ind{\budget, \weight}$ contains complete patterns only, and it represents $\sol_x\ind{\budget, \weight}$.
\end{lemma}
\begin{proof}
    We show this by a bottom-up induction over $\syntaxtree$.
    For an introduce node $i(v)$, it holds that $\sol_x\ind{\budget, \weight}$ is either a singleton or empty. It follows directly from the definitions of both families, that $\rep_x\ind{\budget, \weight} = R_p$ whenever $\sol_x\ind{\budget, \weight} = \{p\}$, and $\rep_x\ind{\budget, \weight}$ is empty otherwise.

    For a relabel node, it holds by induction hypothesis, that $\rep_{x'}\ind{\budget, \weight}$ represents $\sol_{x'}\ind{\budget,\weight}$. Hence, it follows from \cref{lem:ops-preserve-rep}, that $\rep_x\ind{\budget,\weight}$ represents $\sol_x\ind{\budget, \weight}$ and by \cref{obs:complete-closed-ops} it contains complete patterns only.
Similarly, for a union node, it holds by induction hypothesis, that $\rep_{x_1}\ind{b, c}$ represents $\sol_{x_1}\ind{b,c}$, and $\rep_{x_2}\ind{b,c}$ represents $\sol_{x_2}\ind{b,c}$ for all values of $b$ and $c$. Hence, it holds by \cref{lem:ops-preserve-rep} that $\rep_x\ind{\budget, \weight}$ represents $\sol_x\ind{\budget, \weight}$, and by \cref{obs:complete-closed-ops} it holds that it contains complete patterns only.

Finally, for a join node, it holds by induction hypothesis, that $\rep_{x'}\ind{\budget, \weight}$ represents $\sol_{x'}\ind{\budget, \weight}$. Hence, by \cref{lem:patadd-preserve-rep}, it holds that $\rep'_x\ind{\budget,\weight}$ represents $\sol_x\ind{\budget, \weight}$. If follows from \cref{cor:family-complete-rep} and by transitivity of representation, that $\rep_x\ind{\budget,\weight}\subseteq \Cp$, and it represents $\sol_x\ind{\budget, \weight}$.
\end{proof}
\begin{corollary}\label{cor:solution-if-rep-consist-0}
    The graph $G$ admits a Steiner tree of size $\budget$ and weight $\weight$, if and only if it holds that $[0]\in\rep_r\ind{\budget, \weight}$.
\end{corollary}
\begin{proof}
    It holds by \cref{lem:complete-consistentcy} that $[0]$ is the only complete pattern consistent with the pattern $[0]$. Hence, it holds by \cref{lem:rep-represent-sol} that $[0]\in \rep_r\ind{\budget, \weight}$ if and only if there exists a pattern $p\in \sol_r\ind{\budget, \weight}$ with $p\sim [0]$. It follows from \cref{lem:solution-if-sol-consist-0} that this is the case if and only if $G$ admits a Steiner tree of size $\budget$ and weight $\weight$.
\end{proof}

\section{Counting representation}\label{sec:rep-count}
\subsection{Actions and action sequences}
\begin{definition}
	Let $\op:\Pat^k\rightarrow \Pat$ be a $k$-ary pattern operation. We define the operation $\exc{\op}:\left(2^{\Pat}\right)^k\rightarrow 2^{\Pat}$ over subsets of $\Pat$ as 
	\[
	\exc{\op}(S_1, \dots, S_k) = \bigdelta\limits_{\substack{(p_1,\dots, p_k)\in\\S_1\times \dots \times S_k}} \{\op(p_1,\dots, s_k)\},
	\]
for all $S_1, \dots, S_k\subseteq \Pat$, and we call it the \emph{exclusive} version of $\op$. We compare this operation to the notion $\op(S_1, \dots, S_k)$ given by the union over all resulting patterns from the operation.
Given two sets of patterns $S_1, S_2 \subseteq \Pat$, we define $S_1\exjoin S_2, S_1\exunion S_2, \exrel{S_1}{i}{j}$ and $\exadd_{i,j}S_1$ as the exclusive version of join, union, relabel and $\patadd_{i,j}$ of $S_1$ (and $S_2$) respectively.
\end{definition}
\begin{lemma}\label{lem:exclusion-preserve-parity}
	Let $\op:\Pat^k\rightarrow \Pat$ be a $k$-ary pattern operation, and $S_1,\dots S_k\subseteq \Pat$.
    It holds for $q\in \Pat$, that
    \[
        |\{p \in \exop{\op}(S_1,\dots S_k)\colon p \sim q\}|\bquiv |\{(p_1,\dots p_k)\in S_1\times\dots\times S_k\colon \op(p_1,\dots p_k)\sim q\}|        
    \]
\end{lemma}
\begin{proof}
    \begin{align*}
        \left|\{p \in \exop{\op}(S_1,\dots S_k)\colon p \sim q\}\right| &\bquiv
        \left|\bigdelta\limits_{\substack{(p_1,\dots p_k)\in S_1\times\dots\times S_k\\ \op(p_1\dots p_k) \sim q}}\{\op(p_1,\dots p_k)\}\right|\\
        &\bquiv\sum\limits_{\substack{(p_1,\dots p_k)\in S_1\times\dots\times S_k\\ \op(p_1\dots p_k) \sim q}}|\{\op(p_1,\dots p_k)\}|\\
        &\bquiv |\{(p_1,\dots p_k)\in S_1\times\dots\times S_k \colon \op(p_1,\dots p_k)\sim q\}|    
    \end{align*}
\end{proof}

For $p\in \Cp$, and $p' = \patadd_{i,j}p$, it holds that either $R_{p'} = p'$, or $R_{p'}$ contains four patterns that result from $p'$ by either forgetting or fixing both labels $i$ and $j$ independently. We can enumerate these patterns as follows

\begin{definition}\label{def:repactions}
    We define a new operation called \emph{actions} $\action\colon \Pat \times [4] \rightarrow \Pat$ as follows: Given a pattern $p \in \Pat$, if $p\in \Cp$, we define $\action(p, \ell)= p$, for all $\ell\in [4]$.
    For $\inc(p)=\{i\}$, we define 
    \begin{itemize}
        \item $\action(p, 1) = \fix(p, i)$,
        \item $\action(p, 2) = \forget(p, i)$.
    \end{itemize}
    We set $\action(p, 3)$ and $\action(p, 4)$ to be undefined in this case.

    If $\inc(p)=\{i, j\}$ for $i < j$, we define 
    \begin{itemize}
        \item $\action(p, 1) = \fix(\fix(p, i), j)$,
        \item $\action(p, 2) = \forget(\fix(p, i), j)$,
        \item $\action(p, 3) = \fix(\forget(p, i), j)$,
        \item $\action(p, 4) = \forget(\forget(p, i), j)$.
    \end{itemize}
    In all other cases we set $\action(p, \ell)$ to be undefined for all $\ell\in [4]$.
\end{definition}

Using this enumeration, we can assign to each pattern in $\rep_x\ind{\budget, \weight}$ a weight, given by the actions taken to build $p$ at each introduce or join node in the subtree of $\syntaxtree_x$.
We use this weights to isolate different representations of the same pattern in $\sol_x\ind{\budget, \weight}$.

\begin{definition}\label{def:action-sequence}
    For $x\in V(\syntaxtree)$ let $\pi:V(\syntaxtree_x) \rightarrow [4]$ be some mapping, such that $\pi(x') \in [2]$ whenever $x'$ is an introduce node, and $\pi(x') = 1$ whenever $x'$ is a union node or a relabel node.
    We call $\pi$ an \emph{action sequence}.
    Let $\D\in \mathbb{N}$ be some value that will be fixed later, and let $\actionf:V(\syntaxtree)\times[4]\rightarrow [\D]$ be a weight function. 
    We define the weight of $\pi$ as
    \[\actionf(\pi)= \sum\limits_{\substack{x'\in V(\syntaxtree_x),\\x'\text{ is a join or introduce node}}}\actionf\big(x',\pi(x')\big).\]

    Given a partial solution $S \subseteq V_x$ where $\term_x\subseteq S$, the action sequence $\pi$ defines a pattern $p^{\pi}_S$, that can be defined recursively over $\syntaxtree_x$, where for an introduce node $i(v)$, let $v$ be the vertex introduced at this node. If $v\notin S$, then we set $p^{\pi}_S=[0]$. Otherwise, we set $p^{\pi}_S = \action(\pat_x(\{v\}), \pi(x))$. For a relabel or union node, $p^{\pi}_S$ results by applying the corresponding operation on the patterns resulting at the children of $x$ in the natural way. For a join node $\mu_x=\add{i}{j}(\mu_{x'})$, let $p' = \patadd_{i,j}p^{\pi'}_S$, where $\pi'$ is the restriction of $\pi$ to the nodes in $V(\syntaxtree_{x'})$. We set $p^{\pi}_S = \action(p', \pi(x))$.
    We say that $\pi$ is an action sequence generating the pattern $p^{\pi}_S$ from $S$ in $G_x$.
\end{definition}

\begin{lemma}\label{lem:action-seq-represent}
    Let $x\in V(\syntaxtree)$, and two values $\budget, \weight$. It holds for $p\in\Cp$ that $p\in \rep_x(\budget, \weight)$ if and only if there exists a partial solution $S\subseteq V_x$ of size $\budget$ and weight $\weight$ with $\term_x\subseteq S$, and an action sequence $\pi$ that generates $p$ from $S$ in $G_x$.
\end{lemma}
\begin{proof}
    We prove the claim by induction over the syntax tree. For an introduce node, the claim holds by \cref{def:rep-family} and \cref{def:action-sequence}.
    For a relabel node $\mu_x = \relabel{i}{j}(\mu_{x'})$, and $p\in \Cp$, it holds that $p\in \rep_x\ind{\budget, \weight}$ if and only if there exists a pattern $p'\in\rep_{x'}\ind{\budget, \weight}$, with $p'_{i\rightarrow j} = p$. By induction hypothesis, this is the case if and only if there exists a partial solution $S$ over $G_{x'}$ of size $\budget$ and weight $\weight$, and an action sequence $\pi:V(\syntaxtree_{x'})\rightarrow [4]$ with $p^{\pi}_S = p'$. Let $\pi' = \pi_{|x\mapsto 1}$. It holds by \cref{def:action-sequence} that $p^{\pi'}_S = p$.

    Similarly for a union node $\mu_x = \mu_{x_1}\union \mu_{x_2}$, and $p\in \Cp$, it holds that $p\in \rep_x\ind{\budget, \weight}$ if and only if there exist two patterns $p_1 \in \rep_{x_1}\ind{b_1, c_1}$, and $p_2\in\rep_{x_2}\ind{b_2, c_2}$, where $b_1+b_2=\budget$, and $c_1+c_2=\weight$, such that $p_1\punion p_2 = p$. By induction hypothesis, this is the case if and only if there exist two partial solutions $S_1 \subseteq V_{x_1}$ of size $b_1$ and weight $w_1$, and $S_2\subseteq V_{x_2}$ of size $b_2$ and weight $c_2$, and two action sequences $\pi_1: V(\syntaxtree_{x_1})\rightarrow [4]$ and $\pi_2: V(\syntaxtree_{x_2})\rightarrow [4]$ with $p^{\pi_1}_{S_1} = p_1$, and $p^{\pi_2}_{S_2} = p_2$. It holds for $\pi = (\pi_1\dot \cup \pi_2)|_{x\mapsto 1}$ and $S = S_1\cup S_2$ that $S$ is a partial solution over $V_x$ of size $\budget$ and weight $\weight$ that contains $\term_x$, and that $p^{\pi}_S = p$.

    Finally, for a join node $\mu_x = \add{i}{j}(\mu_{x'})$. Let $p\in \Cp$. Then it holds that $p \in \rep_x\ind{\budget, \weight}$ if and only if there exists $p' \in \rep'_x\ind{\budget, \weight}$ where either $\inc(p')=\emptyset$ and $p'=p$, in this case we set $\ell = 1$, or $\inc(p')=\{i,j\}$, and there exists $\ell\in[4]$ with $\action(p', \ell) = p$. However, this holds if and only if there exists a pattern $p'' \in \rep_{x'}\ind{\budget, \weight}$ with $\patadd_{i,j}p'' = p'$. By induction hypothesis, this is the case if and only if there exists a partial solution $S\subseteq V_{x'} = V_x$ of size $\budget$ and weight $\weight$, and an action sequence $\pi'$ over $V_{x'}$ with $p^{\pi'}_{S} = p''$. It holds for $\pi = \pi'_{|x\mapsto \ell}$ that $p^{\pi}_S = \action(\patadd(p'', \ell), \ell) = \action(p', \ell) = p$.
\end{proof}
\subsection{Counting action sequences}
\begin{definition}\label{def:drep-family}
    For $x\in V(\syntaxtree), \budget \in [n], \weight \in [n\cdot \W], \actionweight \in [|V(\syntaxtree)|\cdot D]$,
    we define the families $\drep_x\ind{\budget, \weight, \actionweight}$ as follows:
    For an introduce node $i(v)$, let $c = \weightf(v)$. For $\ell\in[2]$, we set
    \[\drep_x\ind{1,c,\actionf(x, \ell)}= \big\{\action\big(\pat_x(\{v\}), \ell\big)\big\},\]
    and set 
    $\drep_x\ind{0,0,\actionf(x,\ell)}$ to $\{[0]\}$ if $v \notin \term$, and to $\emptyset$ otherwise. For all other values of $\budget, \weight$ and $\actionweight$, we set $\drep_x\allind= \emptyset$.

    For a relabel node $\mu_x = \relabel{i}{j}(\mu_{x'})$ we set
    \[
        \drep_x\ind{\budget,\weight,\actionweight}=  \exrel{\big(\drep_{x'}\ind{\budget,\weight,\actionweight}\big)}{i}{j}.
    \]
    For a union node $\mu_x = \mu_{x_1}\union \mu_{x_2}$ we set
    \[
        \drep_x\ind{\budget,\weight,\actionweight} =
        \bigdelta\limits_{\substack{b_1+b_2 = \budget\\c_1+c_2=\weight\\d_1+d_2=\actionweight}}\big(\drep_{x_1}\ind{b_1,c_1,d_1}\exunion\drep_{x_2}\ind{b_2, c_2,d_2}\big).
    \]
    Finally, for a join node $\mu_x = \add{i}{j}(\mu_{x'})$, we define the families
    \[
        \drep'_x\ind{\budget,\weight,\actionweight}= \exadd_{i,j}\big(\drep_{x'}\ind{\budget,\weight,\actionweight}\big).
    \]
    Now we set
    \[
        \drep_x\ind{\budget,\weight, \actionweight}= \bigdelta\limits_{\ell\in[4]}
            \exac\big(\drep'_{x}\ind{\budget,\weight,\actionweight - \actionf(x, \ell)}, i\big).
    \]
\end{definition}
\begin{lemma}\label{lem:drep-count-actionseq}
    For $x\in V(\syntaxtree),$ and all values $\budget, \weight$ and $\actionweight$, it holds for $p\in \Cp$ that $p\in \drep_x\allind$ if and only if there exist odd many pairs $(S, \pi)$ where $S \subseteq V_x$ with $\term_x \subseteq S, |S| = \budget, \weightf(S) = \weight$ and $\pi$ is an action sequence of weight $\actionweight$ generating $p$ from $S$ in $G_x$.
\end{lemma}
\begin{proof}
    Let $r\nodeind(p)$ be the number (modulo 2) of pairs $(S, \pi)$ meeting all conditions stated in the lemma.
    We prove by induction over $\syntaxtree$ that 
    \[p\in \drep_x\allind \iff r\nodeind(p) \bquiv 1.\]
    Equivalently, we show that
    \[r\nodeind(p) \bquiv [p\in \drep_x\allind].\]
    The claim holds trivially for an introduce node.
    For a relabel node $\mu_x = \relabel{i}{j}(\mu_{x'})$, it holds that 
\[        \drep_x\ind{\budget,\weight,\actionweight}
        =  \exrel{\big(\drep_{x'}\ind{\budget,\weight,\actionweight}\big)}{i}{j}
        = \bigdelta\limits_{p\in \drep_{x'}\allind} \{p_{i\rightarrow j}\}.
    \]
    Each sequence $\pi$ generating $p$ from $S$ in $G_x$ results from some sequence $\pi'$ generating some pattern $p'$ from $S$ in $G_{x'}$, such that $p'_{i\rightarrow j} = p$, and vice versa, where it holds that $\pi = \pi'|_{x\mapsto 1}$. Hence, it holds that
    \begin{align*}
        r\nodeind(p) 
        &\bquiv \sum\limits_{p'_{i\rightarrow j} = p} r\ind{x', \budget, \weight, \actionweight}(p')\\
        &\bquiv \sum\limits_{p'_{i\rightarrow j} = p} [p'\in \drep_{x'}\allind]\\
        &\bquiv |\{p'\in\drep_{x'}\allind\colon p'_{i\rightarrow j}=p\}|,
    \end{align*}
    where the second congruence follows from the induction hypothesis. Hence, it holds that
    \begin{align*}
        &p\in \drep_x\ind{\budget,\weight,\actionweight} \iff\\
        &p\in \bigdelta\limits_{p\in \drep_{x'}\allind} \{p_{i\rightarrow j}\} \iff\\
        &|\{p'\in \drep_{x'}\allind\colon p'_{i\rightarrow j} = p\}| \bquiv 1 \iff\\
        &r\nodeind(p)  \bquiv 1.
    \end{align*}

    For a union node $\mu_x = \mu_{x_1}\union \mu_{x_2}$ it holds
    \begin{align*}
	    \drep_x\ind{\budget,\weight,\actionweight} 
	&=\bigdelta\limits_{\substack{b_1+b_2 = \budget\\c_1+c_2=\weight\\d_1+d_2=\actionweight}}\big(\drep_{x_1}\ind{b_1,c_1,d_1}\exunion\drep_{x_2}\ind{b_2, c_2,d_2}\big)\\
	&=\bigdelta\limits_{\substack{b_1+b_2 = \budget\\c_1+c_2=\weight\\d_1+d_2=\actionweight}}
	\bigdelta\limits_{\substack{p_1\in\drep_{x_1}\ind{b_1, c_1,d_1}\\ p_2\in\drep_{x_2}\ind{b_2,c_2,d_2}}} \{p_1 \punion p_2\}.
    \end{align*}
    For each action sequence $\pi$ of weight $\actionweight$ generating $p$ from a solution $S$ of size $\budget$ and weight $\weight$ in $G_x$, there exists two action sequences $\pi_1, \pi_2$ with $\pi = (\pi_1\dot \cup \pi_2)_{x\mapsto 1}$ and vice versa, where $\pi_1$ has weight $d_1$, and it generates a pattern $p_1$ from a solution $S_1$ over $G_{x_1}$ of size $b_1$ and weight $c_1$, and $\pi_2$ has weight $d_2$, and it generates a pattern $p_2$ from a solution $S_2$ over $G_{x_2}$ of size $b_2$ and weight $c_2$, such that $p_1\punion p_2 = p$, $b_1+b_2=\budget$, $c_1+c_2=\weight$, $d_1+d_2 = \actionweight$. Hence, it holds that
    \begin{align*}
        r\nodeind(p)
        &\bquiv \sum\limits_{\substack{b_1+b_2=\budget\\c_1+c_2=\weight\\d_1+d_2=\actionweight}}
	\sum\limits_{\substack{p_1,p_2\in \Cp\\p_1 \punion p_2 = p}}
	r\ind{x_1, b_1, c_1, d_1}(p_1)\cdot r\ind{x_2, b_2, c_2, d_2}(p_2)\\
        &\bquiv \sum\limits_{\substack{b_1+b_2=\budget\\c_1+c_2=\weight\\d_1+d_2=\actionweight}}
	\sum\limits_{\substack{p_1,p_2\in \Cp\\p_1 \punion p_2 = p}}
	\big[p_1\in\drep_{x_1}\ind{b_1,c_1,d_1}\big]\cdot\big[p_2\in\drep_{x_2}\ind{b_2,c_2,d_2}\big]\\
        &\bquiv \sum\limits_{\substack{b_1+b_2=\budget\\c_1+c_2=\weight\\d_1+d_2=\actionweight}}
	\sum\limits_{\substack{p_1\in\drep_{x_1}\ind{b_1,c_1,d_1}\\p_2\in\drep_{x_2}\ind{b_2,c_2,d_2}}}
	[p_1 \punion p_2 = p].\\
    \end{align*}
    Hence, it holds that
    \begin{align*}
        &p\in \drep_x\ind{\budget,\weight,\actionweight} \iff\\
	&p \in \bigdelta\limits_{\substack{b_1+b_2 = \budget\\c_1+c_2=\weight\\d_1+d_2=\actionweight}}
	\bigdelta\limits_{\substack{p_1\in\drep_{x_1}\ind{b_1, c_1,d_1}\\ p_2\in\drep_{x_2}\ind{b_2,c_2,d_2}}} \{p_1 \punion p_2\}\iff\\
	&\sum\limits_{\substack{b_1+b_2 = \budget\\c_1+c_2=\weight\\d_1+d_2=\actionweight}}
	\sum\limits_{\substack{p_1\in\drep_{x_1}\ind{b_1, c_1,d_1}\\ p_2\in\drep_{x_2}\ind{b_2,c_2,d_2}}} [p_1 \punion p_2 = p]\bquiv 1\iff\\
        &r\nodeind(p)  \bquiv 1.
    \end{align*}
    
    Finally, for a join node $\mu_x = \add{i}{j}(\mu_{x'})$. We define 
    \[r'\nodeind(p) = \sum\limits_{\substack{p'\in \Cp\\ \patadd_{i,j}p' = p}} r\ind{x',\budget,\weight,\actionweight}(p').\]
First, we show that
\[p\in \drep'_x\ind{\budget,\weight,\actionweight} \iff r'\nodeind(p)  \bquiv 1.\]
    It holds that
	\[
	\drep'_x\ind{\budget,\weight,\actionweight}
        =  \exadd_{i,j}{\big(\drep_{x'}\ind{\budget,\weight,\actionweight}\big)}
        = \bigdelta\limits_{p\in \drep_{x'}\allind} \{\patadd_{i, j} p\}.
	\]
	On the other hand, we have
    \begin{align*}
        r'\nodeind(p) 
        &\bquiv \sum\limits_{\patadd_{i,j}p' = p} r\ind{x', \budget, \weight, \actionweight}(p')\\
	&\bquiv \sum\limits_{\patadd_{i,j}p' = p} \big[p'\in \drep_{x'}\allind\big]\\
        &\bquiv |\{p'\in\drep_{x'}\allind\colon \patadd_{i,j}p'=p\}|,
    \end{align*}
    where the second congruence follows from the induction hypothesis.
    It follows that
    \begin{align*}
        &p\in \drep'_x\ind{\budget,\weight,\actionweight} \iff\\
	&p\in \bigdelta\limits_{p\in \drep_{x'}\allind} \{\patadd_{i,j}p\} \iff\\
	&|\{p'\in \drep_{x'}\allind\colon \patadd_{i,j}p' = p\}| \bquiv 1 \iff\\
        &r'\nodeind(p)  \bquiv 1.
    \end{align*}
    Now we prove the original claim. It holds that
    \begin{align*}	
	\drep_x\ind{\budget,\weight,\actionweight}
	&=\bigdelta\limits_{\ell\in[4]}
            \exac\big(\drep'_{x}\ind{\budget,\weight,\actionweight - \actionf(x, \ell)}, \ell\big)\\
	&=\bigdelta\limits_{\ell\in[4]}\bigdelta\limits_{p\in\drep'_x\ind{\budget, \weight, \actionweight-\actionf(x, \ell)}}\{\action(p,\ell)\}.
    \end{align*}
    From the definition of $r$, it holds that
    \begin{align*}
        r\nodeind(p)
	&\bquiv \sum\limits_{\ell\in [4]}
	\sum\limits_{\substack{p'\in \Pat\\ \action(p',\ell) = p}}
	r'\ind{x, \budget, \weight, \actionweight-\actionf(x, \ell)}(p)\\
	&\bquiv \sum\limits_{\ell\in [4]}
	\sum\limits_{\substack{p'\in \Pat\\ \action(p',\ell) = p}}
	\big[p'\in\drep'_{x}\ind{\budget,\weight,\actionweight - \actionf(x, \ell)}\big]\\
	&\bquiv \sum\limits_{\ell\in [4]}
	\sum\limits_{p'\in\drep'_{x}\ind{\budget,\weight,\actionweight - \actionf(x, \ell)}}
	[ \action(p',\ell) = p],
    \end{align*}
    where the second congruence follows from the previous claim.
    Finally, we get
    \begin{align*}
        p\in \drep_x\ind{\budget,\weight,\actionweight} 
	&\iff p\in \bigdelta\limits_{\ell\in[4]}\bigdelta\limits_{p'\in \drep'_{x}\ind{\budget,\weight,\actionweight-\actionf(x,\ell)}}
	\{\action(p', \ell)\}\\
	&\iff \sum\limits_{\ell\in [4]}
	\sum\limits_{p'\in\drep'_{x}\ind{\budget,\weight,\actionweight - \actionf(x, \ell)}} 
	[ \action(p',\ell) = p]\bquiv 1\\
	&\iff r\ind{x, \budget,\weight,\actionweight}(p)  \bquiv 1.
    \end{align*}
\end{proof}
\begin{corollary}\label{lem:no-sol-no-drep}
    If $G$ does not admit a Steiner tree of size $\budget$, then it holds for all values $\weight\in [\W\cdot n]$, $\actionweight\in[\D\cdot |V(\syntaxtree)|]$ that $[0]\notin\drep_r\allind$.
\end{corollary}
\begin{proof}
    Assume there exists two values $\weight$ and $\actionweight$ such that $[0]\in\drep_r\allind$, then by \cref{lem:drep-count-actionseq} there exists a pair $(S, \pi)$ with $\term_x \subseteq S$, $|S|=\budget, \weightf(S) = \weight$ and $\actionf(\pi) = \actionweight$, such that $\pi$ generates $[0]$ from $S$ in $G_r$. However, by \cref{lem:action-seq-represent}, this is the case if and only if $[0]\in \rep_r\ind{\budget, \weight}$, which by \cref{cor:solution-if-rep-consist-0} is the cases if and only if $G$ admits a Steiner tree of size $\budget$.
\end{proof}
\subsection{Isolation lemma}
Now we fix $\W = (2+\sqrt 2) |V|$, and $\D = 4(2+\sqrt 2) |V(\syntaxtree)|$. We choose $\weightf\in [\W]^V$ and $\actionf \in [\D]^{V(\syntaxtree)\times [4]}$ uniformly and independently at random.
We use the isolation lemma to show that, if a Steiner tree of size $\budget$ exists, then there exist two positive integers $\weight$ and $\actionweight$, such that, with high probability, there exists a unique solution $S$ of size $\budget$ and weight $\weight$, and a unique action sequence of weight $\actionweight$ that generates the pattern $[0]$ from $S$.
\begin{definition}
    A function $\omega:U\rightarrow \mathbb{Z}$ isolates a set family $\mathcal{F}\subseteq \pow{U}$ if there exists a unique $S'\in\mathcal{F}$ with $\omega(S')=\min_{S\in \mathcal{F}}\omega(S)$.
\end{definition}
\begin{lemma}[\cite{MulmuleyVV87,DBLP:journals/talg/CyganNPPRW22}]\label{lem:iso}
    Let $\mathcal{F}\subseteq \pow{U}$ be a set family over a universe $U$ with $|\mathcal{F}|>0$, and let $N>|U|$ be an integer. For each $u\in U$, choose a weight $\omega(u)\in \{1,2,\dots N\}$ uniformly and independently at random. Then it holds that $P[\omega \text{ isolates } \mathcal{F}]\geq 1-|U|/N$.
\end{lemma}
\begin{lemma}
	If a Steiner tree of size $\budget$ exists, let $\weight$ be the smallest weight of such a tree. Then there exists a unique Steiner tree of size $\budget$ and weight $\weight$ with probability at least $1- \frac{1}{2 + \sqrt 2}$.
\end{lemma}
\begin{proof}
	We apply \cref{lem:iso}, with $U = V$, and $\mathcal{F}$ the family of all Steiner trees of size $\budget$. It follows that there exists a unique solution of weight $\weight$ with probability at least
	\[1-\frac{|V|}{|V|(2+\sqrt 2)}= 1-\frac{1}{2+\sqrt 2}\]
\end{proof}
\begin{lemma}\label{lem:iso-rep}
    Let $S$ be a Steiner tree in $G$, and $\pi$ be a minimum weight action sequence generating $[0]$ from $S$ in $G$. Then $\pi$ is unique with probability at least $1-\frac{1}{2+\sqrt 2}$.
\end{lemma}
\begin{proof}
    Such a sequence exists by \cref{cor:solution-if-rep-consist-0} and \cref{lem:action-seq-represent}. We show that such a sequence is unique with the given probability. Let $\Pi$ be the family of all action sequences generating $[0]$ from $S$ in $G$. Let $X$ be the set of all "introduce" and "join" nodes in $V(\syntaxtree)$, and let $U = X\times[4]$.
    We define the family $\mathcal{F}\subseteq \pow{U}$ as $\mathcal{F}=\{f_{\pi}\colon \pi \in \Pi\}$, where we define $f_{\pi}= \{(x, \pi(x))\colon x \in X\}$. We also define the weighting function $d':U\rightarrow [\D]$ with $d'((x, \ell)) = \actionf(x,\ell)$.

    Since for any action sequence $\pi$ it holds that $\pi(x) = 1$ for $x\notin X$, it holds for any two different action sequences $\pi_1, \pi_2 \in \Pi$ that there exists $x \in X$ with $\pi_1(x)\neq\pi_2(x)$. Hence, the mapping $\pi\mapsto f_{\pi}$ is a bijection from $\Pi$ to $\mathcal{F}$. It holds from \cref{def:action-sequence} for all $\pi \in \Pi$ that $d'(f_{\pi}) = \actionf(\pi)$. It holds by \cref{lem:iso}, that $d'$ isolates $\mathcal{F}$ with probability at least $1-\frac{1}{2+\sqrt 2}$, and hence, $\actionf$ isolates $\Pi$ with at least this probability.
\end{proof}
\begin{lemma}\label{lem:iso-total-prob}
If $G$ admits a Steiner tree of size $\budget$, then with probability at least $1/2$ it holds that $[0] \in \drep_r\ind{\budget, \weight, \actionweight}$ for some values of $\weight$ and $\actionweight$.
\end{lemma}
\begin{proof}
    Let $\weight$ be the minimum weight of a Steiner tree of size $\budget$ in $G$. Moreover, let $\actionweight$ be the minimum weight of an action sequence generating $[0]$ from some Steiner tree of size $\budget$ and weight $\weight$. Let $A$ be the event that $G$ admits a unique Steiner tree of size $\budget$ and weight $\weight$. Let $B$ be the event that there exists a unique action sequence of weight $\actionweight$ generating $[0]$ from some Steiner tree of size $\budget$ and weight $\weight$ in $G$. Then it holds that $[0] \in \drep_r\ind{\budget,\weight,\actionweight}$ with probability at least
    \begin{align*}
        \pr[A\cap B] &= \pr[A]\cdot \pr[B|A]\\
        &\geq \left(1-\frac{1}{2+\sqrt 2}\right)\cdot\left(1-\frac{1}{2+\sqrt 2}\right)\\
        &= \left(\frac{1+\sqrt 2}{2+\sqrt 2}\right)^2\\
        &= \left(\frac{(1+\sqrt 2)(2-\sqrt 2)}{2}\right)^2\\
        &= \left(\frac{2-2+2\sqrt 2 - \sqrt 2}{2}\right)^2\\
        &= \left(\frac{\sqrt 2}{2}\right)^2 = \frac{1}{2}.
    \end{align*}
\end{proof}

\section{Parity-representation}\label{sec:parity-rep}
Now we define a new kind of representation, called the parity-representation. Parity-representation allows us to count patterns (modulo 2) that are consistent with a specific pattern. However, this representation is restricted to complete patterns only, and hence, we aim to isolate a unique representation of a valid solution with high probability. This would imply a single-sided error randomized algorithm for the decision version of the \Stp\ problem in running time $\ostar(3^{\cw(G)})$.

\begin{definition}\label{def:cs-pat}
    We define the family of $CS$-patterns $\CSP(U) \subseteq \Cp(U)$ as the family of complete patterns containing only a zero-set and singletons, i.e.\ 
    \[\CSP(U) =\Big\{ \big\{X\cup \{0\}\big\} \cup \big\{\{u\}\colon u\in Y\big\}\colon X\subseteq Y\subseteq U\Big\}.\]
\end{definition}
\begin{lemma}\label{lem:num-of-cs-pats}
    For each ground set $U$ it holds that $|\CSP(U)| =  3^{|U|}$.
\end{lemma}
\begin{proof}
    It holds that $|\CSP|=\sum\limits_{Y\subseteq U}2^{|Y|} = 3^{|U|}$.
\end{proof}

\begin{observation}\label{obs:cs-closed-ops}
    For $p, q \in \CSP$ and all $i,j\in U$, it holds that $p_{i\rightarrow j}, p \punion q \in \CSP$.
\end{observation}

In the following, we show that $\CSP$ builds a row basis of $\cmat$ (\cref{sec:technicalcontribution}) over $\bin$. Note that $\CSP$ is a row basis of $\cmat$, if and only if for each $p\in \Cp$ there exists $S \subseteq \CSP$ such that for each $q\in \Cp$ it holds that $p\sim q$ if and only if $|\{r\in S\colon r\sim q\}|$ is an odd number. We call $S$ a basis representation (or \emph{parity-representation}) of $p$.

\begin{definition}
    Given two sets of patterns $D, C \subseteq \Pat$, we say that 
    $C$  \emph{parity-represents} $D$, if and only if for each $q\in \Cp$ it holds
    \[|\{p\in C\colon p\sim q\}|\bquiv|\{p\in D\colon p\sim q\}|.\]
	For a set of patterns $C \subseteq \Pat$ and a single pattern $p\in \Pat$, we say that $C$ parity-represents $p$ if it holds that $C$ parity-represents $\{p\}$.
\end{definition}

\begin{observation}\label{obs:par-rep-is-equiv}
Parity-representation is an equivalence relation.
\end{observation}

\begin{lemma}\label{lem:mutex-preserve-parity-rep}
    Let $S,S',T,T' \subseteq \Pat$ with $S$ parity-represents $S'$ and $T$ parity-represents $T'$. Then it holds that $S\triangle T$ parity-represents $S'\triangle T'$.
\end{lemma}
\begin{proof}
    For an arbitrary $q\in \Cp$, it holds that
    \begin{align*}
        |\{p\in S\triangle T\colon p\sim q\}|
        &\bquiv |\{p\in S\colon p\sim q\}\triangle\{p\in T\colon p\sim q\}|\\
        &\bquiv |\{p\in S\colon p\sim q\}|+|\{p\in T\colon p\sim q\}|\\
        &\bquiv |\{p\in S'\colon p\sim q\}|+|\{p\in T'\colon p\sim q\}|\\
        &\bquiv |\{p\in S'\colon p\sim q\}\triangle\{p\in T'\colon p\sim q\}|\\
        &\bquiv |\{p\in S'\triangle T'\colon p\sim q\}|\\
    \end{align*}
\end{proof}
\begin{corollary}\label{cor:partial-parity-rep}
    Let $C \subseteq \Cp$ be a family of complete patterns, and $S\subseteq C$. Let $S'$ be a parity-representation of $S$, and $C'=(C\setminus S)\triangle S'$. Then $C'$ is a parity-representation of $C$.
    \end{corollary}
    \begin{proof}
        Let $\overline S = C\setminus S$. It holds that $C' = \overline S \triangle S'$. Since $S'$ parity-represents $S$ by definition, it holds by \cref{lem:mutex-preserve-parity-rep} that $C'$ parity-represents $\overline S \triangle S = C$.
\end{proof}

\begin{definition}
    We say that a set of patterns $S \subseteq \Pat$ \emph{jointly dominates} a pattern $p \in \Pat$ over a family of patterns $\Pat'$ (denoted by $S\geq_{\Pat'} p$), if for each pattern $r\in \Pat'$ with $p\sim r$, there exists $q\in S$ such that $q \sim r$. We say that $S$ \emph{totally covers} $p$ over $\Pat'\subseteq \Pat$ if for each pattern $r \in \Pat'$ it holds that $p\sim r$, if and only if $q\sim r$ for each $q\in S$.    
\end{definition}

\begin{lemma}\label{lem:joint-domination}
    Let $X,Y,Z\subseteq U$ be three label sets, and $p\in \Cp$. Let
    \begin{itemize}
        \item $p_0 = p\cup \{X\cup Y, Z\}$,
        \item $p_1 = p\cup \{Y\cup Z, X\}$,
        \item $p_2 = p\cup \{X\cup Z, Y\}$.
    \end{itemize}
    Then $\{p_1, p_2\}$ jointly dominates $p_0$ over $\Pat$.
\end{lemma}

\begin{proof}
    Let $q\in \Pat$ with $p_0\sim q$. Let $W_0$ be a $(p_0,q)$-alternating walk from $Z$ to $X\cup Y$. We assume that the walk visits $X\cup Y$ exactly once, since otherwise we can cut the part after the first visit. Let $Z'$ be the set preceding $X\cup Y$ on the walk. If $Z'\cap X \neq \emptyset$, we claim that $p_1 \sim q$. Otherwise it must hold that $Z'\cap Y \neq \emptyset$. In that case we claim that $p_2\sim q$. We prove the former one, the latter follows by symmetry.

    Assuming $Z'\cap X \neq \emptyset$, we show that there exists a $(p_1,q)$-alternating walk between $Y\cup Z$ and each set of $S \in p_1\cup q$. This is clearly the case for $S=Y\cup Z$. Let $W$ be the walk resulting from $W_0$ by replacing $Z$ with $Y\cup Z$, and replacing $X\cup Y$ with $X$. Then $W$ is a $(p_1,q)$-alternating walk between $Y\cup Z$ and $X$.
    
    Let $S \in (p_1\cup q)\setminus \{Y\cup Z, X\}$. Then there exists a $(p_0,q)$-alternating walk $P_0$ between $Z$ and $S$. We assume that $P_0$ contains $X\cup Y$ at an odd position at most once by removing the part between the first and the last such occurrences. We build the walk $P$ from $P_0$ as follows: First we replace $Z$ with $Y\cup Z$. Moreover, if $X\cup Y$ appears on the walk, let $T$ be the set following $X\cup Y$. If $T\cap Y\neq \emptyset$, we remove the part of the walk between $Z\cup Y$ (originally $Z$) and $T$ (excluding the endpoints). Otherwise, it must hold that $T \cap X \neq \emptyset$. In this case, we replace the part between $Y\cup Z$ and $T$ with $W$. It holds that $P$ is a $(p_1,q)$-alternating walk between $Y\cup Z$ and $S$ for any choice of $S$.
\end{proof}

\begin{remark}
    If follows from \cref{lem:joint-domination} by symmetry that any two of $p_0, p_1, p_2$ jointly dominate the third of them over $\pat$.
\end{remark}

\begin{lemma}\label{lem:set-parity-rep}
    Given a pattern $p\in \Cp \setminus \CSP$. Let $S\in p \setminus \{Z_p\}$ be an inclusion-wise minimal set, such that $|S|>1$.
    Let $A = A^S_p$ be the family of patterns defined from $p$ as follows
    \[A^S_p = \big\{(p\setminus \{Z_p, S\} )\cup (Z_p \cup S') \colon S'\subset S \big\}.\]
    Then $A$ is a parity-representation of $p$ of size at most $2^{|S|}$.
\end{lemma}

\begin{proof}
    Let $p' = p \setminus \{S, Z_p\}$.
    Let $u\in S$ be minimal, and $S' = S\setminus\{u\}$. We define $C^S_p = \{p_0, p_1, p_2\}$, where
    \begin{alignat*}{3}
        p_0 =&& p' &\cup \{Z_p\cup\{u\},S'\},& &\\
        p_1 =&& p' &\cup \{Z_p\cup S',\{u\}\} &=& p' \cup \{Z_p\cup S'\},\\
        p_2 =&& p' &\cup \{Z_p,S',\{u\}\} &=&p' \cup \{Z_p,S'\}.
    \end{alignat*}
    The second equalities in both the second and the third lines follow from the fact that $\sing(p)=\lbs(p)$, and hence, $\{u\} \in p$. We claim that the set $C^S_p$ parity-represents $p$.
    First note that by \cref{lem:joint-domination} any two of $p, p_0, p_1$ jointly dominate the third of them over $\Pat$ (and hence, over $\Cp$).
    
    We claim that the set $\{p, p_0, p_1\}$ totally covers $p_2$ over $\Cp$. Assuming this is the case, let $q\in \Cp$. If $p\sim q$ holds, then either $q\sim p_2$, but then $q$ must be consistent with both $p_0$ and $p_1$, or $q$ is not consistent with $p_2$, and hence, $q$ is consistent with exactly one of $p_0$ and $p_1$ (at least one, since they jointly dominate $p$, but not both, since this would imply that $q\sim p_2$). In this case it holds that 
    \[|\{r\in C^S_p\colon r\sim q\}| \in \{1,3\}.\]
    On the other hand, if $q\not\sim p$, then $q\not \sim p_2$, and if $q$ is consistent to one of $p_0$ and $p_1$ then it must be consistent to both of them.
    Hence, it holds that 
    \[|\{r\in C^S_p\colon r\sim q\}| \in \{0,2\},\]
    which shows that $C^S_p$ parity-represents $p$.

    Now we show that the $\{p, p_0, p_1\}$ totally covers $p_2$.
    It follows from \cref{lem:adding-subset-dom} and \cref{lem:removing-element-dom} that each of these three patterns dominates $p_2$. We show for all $q\in \Cp$, that if $q$ is consistent with all $p,p_0$ and $p_1$, then $q\sim p_2$ holds as well. Note that if a $(p_2,q)$-alternating walk from $\{u\}$ to $S'$ exists in $p_2$, then $p_2 \sim q$ follows from $p \sim q$, since $p$ results from $p_2$ by unifying these two sets.
    
    However, if such a walk doesn't exist, then there exists no $(p_2,q)$-alternating walk from $Z_p$ to $\{u\}$, or there exists no $(p_2,q)$-alternating walk from $Z_p$ to $S'$, since otherwise, we can combine these two walks into a $(p_2,q)$-alternating walk from $\{u\}$ to $S'$. The former case implies that there exists no $(p_1,q)$-alternating walk from $Z_p\cup S'$ to $\{u\}$, which contradicts the assumption that $p_1\sim q$, while the latter implies that there exists no $(p_0,q)$-alternating walk from $Z_p\cup\{u\}$ to $S'$ which contradicts the assumption that $p_0\sim r$. Hence, a $(p_2,q)$-alternating walk from $\{u\}$ to $S'$ must exist, and hence, $p_2\sim r$ holds.

Now let $S = \{u_1,\dots u_r\}$ where $u_1, \dots u_r$ are ordered. Let $S_i = \{u_1, \dots u_i\}$ and $\overline{S_i} = S\setminus S_i$. For a pattern $q$, let $S_q \subseteq S$ be the only set in $q$ that is a subset of $S$ of size at least two, if $q$ contains exactly one such set, or $S_q$ is undefined otherwise. If $S_q$ is defined, let $u$ be the minimum element in $S_q$. We define $S'_q = S_q\setminus\{u\}$. If $S_q$ is undefined, then $S'_q$ is undefined as well. Moreover, let $C_q = C_q^{S_q}$ if $S_q$ is defined, or $C_q = \{q\}$ otherwise. We define $A_0 = \{p\}$. We define $A_{i+1}$ from $A_i$ as follows
\[A_{i+1}=\bigdelta\limits_{q\in A_i}C_q.\]
Using this definition, we claim that for each $i\in[r]$ it holds that
\[A_i = \bigcup_{S'\subseteq S_i} \left(p'\cup \{Z_p\cup S', \overline{S_i}\}\right)
				  \cup \bigcup_{S'\subset S_i} \left(p' \cup \{Z_p\cup S' \cup \overline{S_i}\}\right).\]
Since $C_p$ parity-represents $p$, it holds by \cref{lem:mutex-preserve-parity-rep} that $A_i$ parity-represents $A_{i-1}$. Hence, by transitivity of parity-representation (\cref{obs:par-rep-is-equiv}) it holds that $A_i$ parity-represents $A_0$. Assuming the claim above is true, it holds for $i = r-1$ that
\begin{align*}
A_{r-1} =& \bigcup\limits_{S'\subseteq S_{r-1}}\left(p' \cup \{Z_p\cup S', \{r\}\}\right) \cup  \bigcup_{S'\subset S_{r-1}} \left(p' \cup \{Z_p\cup S' \cup \{r\}\}\right)\\
=&\bigcup_{S'\subset S} \left(p'\cup \{Z_p \cup S' \}\right) = A.
\end{align*} 
Therefore, assuming the claim, it holds that $A = A_{r-1}$ parity-represents $A_0 = \{p\}$, which completes the proof.

Now we prove the aforementioned claim by induction over $i\in \{0,\dots i-1\}$.
For $i=0$, the claim holds trivially, since $S_0 = \emptyset$ does not admit any subset. Now assume the correctness for some fixed value $i$. We prove the claim for $i+1$.
Let $L=\bigcup_{S'\subseteq S_i} p' \cup \{Z_p\cup S', \overline{S_i}\}$ and $R = \bigcup_{S'\subset S_i} p' \cup \{Z_p\cup S' \cup \overline{S_q}\}$. It holds that $L$ and $R$ are disjoint, and by induction hypothesis, it holds that $A_i = L\dot{\cup} R$. Hence,
\[A_{i+1} = \bigdelta\limits_{q\in L}C_q \triangle \bigdelta\limits_{q \in R}C_q.\]
Note that $C_q = \{q\}$ for all $q\in R$. For $q\in L$ it holds that $C_q = C^{S_q}_q$, where $S_q = \overline{S_i}$ and $u_{i+1}$ is minimal in $S_q$, and hence, $S'_q = \overline{S_{i+1}}$. It follows that,
\begin{alignat*}{2}
    A_{i+1} = &\bigdelta\limits_{q\in L}\Big\{&&(q\setminus \big\{S_q, Z_q\big\})\cup \big\{Z_q\cup \{u_{i+1}\}, S'_q\big\},\\
    & &&(q\setminus\big\{S_q, Z_q\big\})\cup \big\{Z_q \cup S'_q, \{u_{i+1}\}\big\},\\
    & &&(q \setminus \big\{S_q, Z_q\big\})\cup \big\{Z_q,S'_q, \{u_{i+1}\}\big\}
    \Big\}\triangle \bigdelta\limits_{q \in R}\{q\}.\\
    = &\bigdelta\limits_{S'\subseteq S_i}\Big\{&&p'\cup\big\{Z_p\cup S'\cup \{u_{i+1}\}, \overline{S_{i+1}}\big\},
    p'\cup\big\{Z_p\cup S' \cup \overline{S_{i+1}}\big\},\\
    & && p'\cup \big\{Z_p\cup S', \overline{S_{i+1}}\big\}\Big\}
    \triangle \bigdelta\limits_{S'\subset S_i}\Big\{p'\cup\big\{Z_p\cup S'\cup \overline{S_{i}}\big\}\Big\}\\
    = &\bigcup\limits_{S'\subseteq S_i}\Big\{&&p'\cup\big\{Z_p\cup S'\cup \{u_{i+1}\}, \overline{S_{i+1}}\big\},
    p'\cup\big\{Z_p\cup S' \cup \overline{S_{i+1}}\big\},\\
    & && p'\cup \big\{Z_p\cup S', \overline{S_{i+1}}\big\}\Big\}
    \cup \bigcup\limits_{S'\subset S_i}\Big\{p'\cup\big\{Z_p\cup S'\cup \overline{S_{i}}\big\}\Big\}\\
    = &\bigcup\limits_{S'\subseteq S_i}\Big\{&&p'\cup\big\{Z_p\cup S'\cup \{u_{i+1}\}, \overline{S_{i+1}}\big\},
    p'\cup \big\{Z_p\cup S', \overline{S_{i+1}}\big\}\Big\} \cup\\
    &\bigcup\limits_{S'\subseteq S_{i}}\Big\{ && p'\cup\big\{Z_p\cup S' \cup \overline{S_{i+1}}\big\}\Big\} \cup
    \bigcup\limits_{S'\subset S_{i}}\Big\{p'\cup\big\{Z_p\cup S'\cup \overline{S_{i}}\big\}\Big\}\\
    = &\bigcup\limits_{S'\subseteq S_{i+1}}\Big\{&&p'\cup\big\{Z_p\cup S'\big\}, \overline{S_{i+1}}\Big\} \cup\\
    &\bigcup\limits_{\substack{S'\subset S_{i+1}\\u_{i+1}\notin S'}}\Big\{ && p'\cup\big\{Z_p\cup S' \cup \overline{S_{i+1}}\big\}\Big\} \cup
    \bigcup\limits_{\substack{S'\subset S_{i+1}\\u_{i+1}\in S'}}\Big\{p'\cup\big\{Z_p\cup S'\cup \overline{S_{i+1}}\big\}\Big\}\\
    = &\bigcup\limits_{S'\subseteq S_{i+1}}\Big\{&&p'\cup\big\{Z_p\cup S'\big\}, \overline{S_{i+1}}\Big\} \cup
    \bigcup\limits_{S'\subset S_{i+1}}\Big\{p'\cup\big\{Z_p\cup S'\cup \overline{S_{i+1}}\big\}\Big\}\\
\end{alignat*}
where the third equality follows from the fact that all patterns that appear in the expression are different. This follows from the fact that we can partition the expression into four families of patterns, each defined by a different zero-set: One containing $u_{i+1}$ but no element of $\overline{S_{i+1}}$, one containing all element of $\overline{S_{i+1}}$ but not $u_{i+1}$, one containing neither, and on containing both. Since we are only interested in computing $A_{r-1}$, it holds that $i\leq r-2$, and hence, $\overline{S_{i+1}}$ is not empty. On the other hand, the patterns appearing in each such family are defined by different sets $S'$, and it is easy to see that they are pairwise different.
\end{proof}

\begin{lemma}\label{lem:pat-parity-rep}
    Given a complete pattern $p\in \Cp$. 
    Let $S_1, \dots S_r$ be the sets in $p$ such that $S_i \neq Z_p$ and $|S_i|\geq 2$ for each $i\in [r]$.
    Let $A_0 = \{p\}$. For $i\in [r]$, we define
    \[ A_i = \bigdelta\limits_{q \in A_{i-1}} A_q^{S_i}.\]
    Let $A_p = A_r$.
    Then it holds that $A_p \subseteq \CSP$, and that it parity-represents $p$.
\end{lemma}

\begin{proof}
    It holds by \cref{lem:set-parity-rep} that $A^{S_i}_q$ parity-represents $q$, and
    by \cref{cor:partial-parity-rep} it follows that $A_i$ parity-represents $A_{i-1}$ for each $i\in[r]$. By transitivity of parity-representation, it holds that $A_r$ parity-represents $p$. Note that for each $i\in [r]$, $A_i$ results from $A_{i-1}$ by replacing each pattern $q$ by the patterns resulting from $q$ by removing $S_i$ and adding all its proper subsets to $Z_q$. Hence, for each $i\in [r]_0$, the sets $S_{i+1}\dots S_r$ are the only sets in $A_i$ that are different from $Z_p$ and have size at least two. This implies that $A_r \subseteq \CSP$.
\end{proof}

\begin{corollary}\label{cor:family-parity-rep}
    For each family of complete patterns $D\subseteq \Cp$, there exists a family of $CS$-patterns $C\subseteq \CSP$ that parity-represents $D$.
\end{corollary}
\begin{proof}
    Let $C=\bigdelta\limits_{p\in D}A_p$, as defined in \cref{lem:pat-parity-rep}. Clearly, it holds that $C \subseteq \CSP$ , and by \cref{cor:partial-parity-rep} that $C$ parity-represents $D$.
\end{proof}

\begin{definition}
    Let $\op:\Pat^k \rightarrow \Pat$ be a $k$-ary operation over patterns, for some positive integer $k$, and $\exop{\op}$ be the exclusive version of it. We say that $\exop{\op}$ \emph{preserves parity-representation}, if for all sets of patterns $P_1, \dots P_k, Q_1, \dots Q_k \in \Pat$, where $P_i$ parity-represents $Q_i$ for all $i \in [k]$, it holds that $\exop{op}(P_1,\dots P_k)$ parity-represents $\exop{\op}(Q_1, \dots Q_k)$.
\end{definition}

\begin{lemma}\label{lem:ops-preserve-parity-rep}
	The operations exclusive relabel $\exrel{S}{i}{j}$, exclusive join $S_1 \exjoin S_2$ and exclusive union $S_1 \exunion S_2$ preserve parity-representation.
\end{lemma}
\begin{proof}
	In the following, let $S_1, S_2, T_1, T_2 \subseteq \Pat$ such that $T_1$ and $T_2$ parity-represent $S_1$ and $S_2$ respectively.
	\begin{itemize}
        \item[]
    Exclusive relabel: let $S = \exrel{(S_1)}{i}{j}$ and $T = \exrel{(T_1)}{i}{j}$. Let $q \in \Cp$. It holds that
	\begin{alignat*}{3}
		&|\{p \in S\colon p \sim q\}| &&\bquiv
		|\{p\in S_1\colon p_{i\rightarrow j}\sim q\}|&\bquiv\\
	        &|\{p \in S_1\colon p \sim q_{j\curvearrowleft i}\}| &&\bquiv
		|\{p \in T_1\colon p\sim q_{j\curvearrowleft i}\}| &\bquiv\\
	        &|\{p \in T_1\colon p_{i\rightarrow j}\sim q\}| &&\bquiv
		|\{p \in T\colon p \sim q\}|,
    \end{alignat*}
    where the first and the last congruences follow from \cref{lem:exclusion-preserve-parity}, while the second and the fourth follow from \cref{lem:rel-is-adding-label-other-side}.
    \item[] 
    Exclusive union: let $S = S_1 \exunion S_2$ and $T = T_1 \exunion T_2$. It holds for $q\in \Cp$ that
    \begin{alignat*}{3}
        &|\{p\in S\colon p\sim q\}| &&\bquiv
        \sum\limits_{p_1\in S_1}|\{p_2\in S_2\colon p_1 \punion p_2 \sim q\}| &\bquiv\\
        &\sum\limits_{p_1\in S_1}|\{p_2\in S_2\colon p_2 \sim p'_1\join q'\}|&&\bquiv
        \sum\limits_{p_1\in S_1}|\{p_2\in T_2\colon p_2 \sim p'_1\join q'\}|&\bquiv\\
        &\sum\limits_{p_1\in S_1}|\{p_2\in T_2\colon p_2 \punion p_1\sim q'\}|&&\bquiv
        \sum\limits_{p_2\in T_2}|\{p_1\in S_1\colon p_1 \punion p_2 \sim q\}| &\bquiv\\
        &\sum\limits_{p_2\in T_2}|\{p_1\in S_1\colon p_1 \sim p'_2 \join q'\}| &&\bquiv
        \sum\limits_{p_2\in T_2}|\{p_1\in T_1\colon p_1 \sim p'_2 \join q'\}| &\bquiv\\
        &\sum\limits_{p_2\in T_2}|\{p_1\in T_1\colon p_1 \punion p_2 \sim q\}| &&\bquiv
        |\{p\in T\colon p\sim q\}|,
    \end{alignat*} 
    where $p'_1, p'_2, q'$ are the patterns resulting from $p_1, p_2$ and $q$ respectively as defined in \cref{cor:union-is-join-other-side}.
    Note that the first and the last congruences follow from \cref{lem:exclusion-preserve-parity}, the second, fourth, sixth and the eighth congruences follow from \cref{cor:union-is-join-other-side}, while the third and the seventh follow from the assumption, that $T_1$ and $T_2$ parity-represent $S_1$ and $S_2$ respectively.
    \item[] Exclusive join: let $S = S_1 \exjoin S_2$ and $T = T_1 \exjoin T_2$. It holds for $q\in \Cp$ that
    \begin{alignat*}{3}
        &|\{p\in S\colon p\sim q\}| &&\bquiv
        \sum\limits_{p_1\in S_1}|\{p_2\in S_2\colon p_1 \join p_2 \sim q\}| &\bquiv\\
        &\sum\limits_{p_1\in S_1}|\{p_2\in S_2\colon p_2 \sim p_1\join q\}| &&\bquiv
        \sum\limits_{p_1\in S_1}|\{p_2\in T_2\colon p_2 \sim p_1\join q\}| &\bquiv\\
        &\sum\limits_{p_1\in S_1}|\{p_2\in T_2\colon p_2 \join p_1\sim q\}| &&\bquiv
        \sum\limits_{p_2\in T_2}|\{p_1\in S_1\colon p_1 \join p_2 \sim q\}| &\bquiv\\ 
        &\sum\limits_{p_2\in T_2}|\{p_1\in S_1\colon p_1 \sim p_2 \join q\}| &&\bquiv 
        \sum\limits_{p_2\in T_2}|\{p_1\in T_1\colon p_1 \sim p_2 \join q\}| &\bquiv\\ 
        &\sum\limits_{p_2\in T_2}|\{p_1\in T_1\colon p_1 \join p_2 \sim q\}| &&\bquiv 
        |\{p\in T\colon p\sim q\}|,&
    \end{alignat*} 
    where the first and the last congruences follow from \cref{lem:exclusion-preserve-parity}, the second, fourth, sixth and the eighth congruences follow from \cref{lem:join-is-join-other-side}, while the third and the seventh follow from the assumption, that $T_1$ and $T_2$ parity-represent $S_1$ and $S_2$ respectively.
    \end{itemize}
\end{proof}

\begin{lemma}\label{lem:patadd-preserve-parity-rep}
    Given two families of complete patterns $C,D\subseteq \Cp$, such that $C$ parity-represents $D$, it follows that $\exadd_{i,j}C$ parity-represents $\exadd_{i,j}D$.
\end{lemma}
\begin{proof}
    Let $q\in \Cp$. It holds for $X\in\{C,D\}$ by \cref{lem:exclusion-preserve-parity} that
    \begin{align*}
        &|\{p\in \exadd_{i,j}X\colon p\sim q\}|
        \bquiv |\{p\in X\colon \patadd_{i,j}p \sim q\}|\bquiv\\
        &\underbrace{|\{p\in X\colon \{i,j\}\subseteq p \land \patadd_{i,j}p \sim q\}|}_{\mathcal{L}_X} +
        \underbrace{|\{p\in X\colon \{i,j\}\not\subseteq p \land \patadd_{i,j}p \sim q\}|}_{\mathcal{R}_X}.\\
        \end{align*}
        It holds by \cref{lem:patadd-is-union-single-otherside} that
        \begin{align*}
            \mathcal{L}_X &\bquiv |\{p\in X\colon \{i,j\}\subseteq\lbs(p)\land p\sim q\punion[ij,i,j]\}|\\
            &\bquiv |\{p\in X\colon p\sim q\punion[ij,i,j]\}|,
        \end{align*}
        where the second congruence holds by \cref{lem:missing-singleton-not-consistent}, since $\{i,j\}\subseteq\sing(q\punion[ij,i,j])$.
        Hence, it holds that
        \begin{align*}
            \mathcal{L}_D &\bquiv |\{p\in D\colon p\sim q\punion[ij,i,j]\}|\\
            &\bquiv|\{p\in C \colon p\sim q\punion[ij,i,j]\}|\bquiv \mathcal{L}_C.
        \end{align*}
        On the other hand, it holds that
        \begin{align*}
            \mathcal{R}_X &= |\{p\in X\colon \{i,j\}\not\subseteq p \land \patadd_{i,j}p \sim q\}|\\
            &= |\{p\in X\colon \{i,j\}\not\subseteq p \land p \sim q\}|.\\
        \end{align*}
        If $\{i, j\} \subseteq \lbs(q)$, then it holds by \cref{lem:complete-consistentcy} that $\mathcal{R}_C = \mathcal{R}_D = 0$. Assuming $\{i, j\}\not \subseteq \lbs(q)$, it holds by \cref{lem:complete-consistentcy} as well that 

        \begin{align*}
        \mathcal{R}_C 
        &= |\{p\in C\colon \{i,j\}\not\subseteq p \land p \sim q\}|\\
        &= |\{p\in C\colon p \sim q\}|\\
        &\bquiv |\{p\in D\colon p \sim q\}|\\
        &= |\{p\in D\colon \{i,j\}\not\subseteq p \land p \sim q\}|\\
        &=\mathcal{R}_D.
        \end{align*}

        Finally, it follows that
    \[
        |\{p\in \exadd_{i,j}D\colon p\sim q\}|\bquiv \mathcal{L}_D+\mathcal{R}_D\bquiv
        \mathcal{L}_C+\mathcal{R}_C\bquiv |\{p\in \exadd_{i,j}C\colon p\sim q\}|.
    \]
\end{proof}

\begin{lemma}\label{lem:ac-preserve-parity-rep}
    Let $D, C\subseteq \Cp$ such that $C$ parity-represents $D$. Let $i,j \in U$ with $i<j$. Let $D' = \exadd_{i,j}D$, and $C' =  \exadd_{i,j}C$. Then it holds for all $\ell\in[4]$ that $\exac(C', \ell)$ parity-represents $\exac(D', \ell)$.
\end{lemma}
\begin{proof}
    \cref{lem:ops-preserve-parity-rep} implies that $C'$ parity-represents $D'$. For each pattern $p\in D' \cup C'$, it holds either $\lbs(p)\setminus \sing(p)=\{i,j\}$, if $\{i, j\} \subseteq \lbs(p)$, or $\lbs(p)=\sing(p)$ otherwise.
    Let $D_0 = \{p \in D\colon \{i, j\}\not \subseteq \lbs(p)\}$, and $\overline{D_0} = D\setminus D_0$. Let $\overline{D_0}' = D'\setminus D_0$. Similarly, let $C_0 = \{p \in C\colon \{i, j\}\not \subseteq \lbs(p)\}$, and $\overline{C_0} = C\setminus C_0$. Let $\overline{C_0}' = C'\setminus C_0$.
    
    For a pattern $q\in \Cp$ and $X\in\{C, D\}$, it holds by \cref{lem:exclusion-preserve-parity} that
    \begin{equation}
        \begin{split}
    |\{p \in \exac(X', \ell)\colon p\sim q\}| \bquiv |\{p\in X'\colon \action(p, \ell) \sim q\}|\\
        \bquiv \underbrace{|\{p\in X_0\colon \action(p,\ell)\sim q \}|}_{\mathcal{L}_X} + \underbrace{|\{p\in \overline{X_0}'\colon \action(p,\ell)\sim q \}|}_{\mathcal{R}_X}.\\
    \end{split}
\end{equation}
It holds that 
\begin{align*}
    \mathcal{L}_X &\bquiv |\{p \in X_0\colon p\sim q\}|\\
    &\bquiv |\{p \in X'\colon \{i, j\}\not \subseteq \lbs(p) \land p\sim q\}|\\
    &\bquiv |\{p \in X \colon \{i, j\}\not \subseteq \lbs(p) \land p\joinifex[i,j] \sim q\}|\\
    &\bquiv |\{p \in X \colon \{i, j\}\not \subseteq \lbs(p) \land p \sim q\}.
\end{align*}
    If $\{i, j\} \subseteq \lbs(q)$, then it holds that $\mathcal{L}_C = \mathcal{L}_D = 0$. Otherwise, assume that $\{i, j\} \not \subseteq \lbs(q)$. It holds that
    \begin{align*}
        \mathcal{L}_D
        &\bquiv |\{p\in D\colon \{i, j\}\not \subseteq \lbs(p) \land p \sim q\}\\
        &\bquiv |\{p\in D\colon p\sim q\}\\
        &\bquiv |\{p\in C\colon p\sim q\}\\
        &\bquiv |\{p\in C\colon \{i, j\}\not \subseteq \lbs(p)\land p \sim q\}\\
        &\bquiv \mathcal{L}_C.
\end{align*}
In both cases it holds that $\mathcal{L}_C \bquiv \mathcal{L}_D$.
Now we show that $\mathcal{R}_C \bquiv \mathcal{R}_D$. 
Let us define the following four conditions
\begin{itemize}
    \item $\ell \in \{1, 2\} \land i\notin \lbs(q)$.
    \item $\ell \in \{3, 4\} \land i\in \lbs(q)$.
    \item $\ell \in \{1, 3\} \land j\notin \lbs(q)$.
    \item $\ell \in \{2, 4\} \land j\in \lbs(q)$.
\end{itemize}
If any of the four conditions is true, it follows that for each $p \in \overline{X_0}'$ that either $\sing(\action(p, \ell))\setminus\lbs(q)\neq \emptyset$ or $\sing(q)\setminus \lbs(\action(p,\ell)) = \emptyset$. Hence, it follows that $\mathcal{R}_C = \mathcal{R}_D = 0$. Assuming none of the conditions is true, it follows that $\action(p,\ell) \sim q$ if and only if $p\sim q$. Hence, it holds that
\begin{align*}
    \mathcal{R}_X
    &\bquiv |\{p\in \overline{X_0}'\colon \action(p, \ell)\sim q\}|\\
    &\bquiv |\{p\in \overline{X_0}'\colon p\sim q\}|\\
    &\bquiv |\{p\in X'\colon p\sim q\}| - |\{p\in X_0\colon p\sim q\}|.\\
\end{align*}
It follows that 
\begin{align*}
    \mathcal{R}_D
    &\bquiv |\{p\in D'\colon p\sim q\}| - |\{p\in D_0\colon p\sim q\}|\\
    &\bquiv |\{p\in C'\colon p\sim q\}| - |\{p\in C_0\colon p\sim q\}|,\\
    &\bquiv \mathcal{R}_C,
\end{align*}
where the second congruence follows from the fact that $C'$ parity-represent $D'$, and $\mathcal{R}_C \bquiv \mathcal{R}_D$.
\end{proof}

Now we define families of partial solutions over $\CSP$ that parity-represent the families $\drep_x\allind$. In the next section, we show that we can efficiently compute these families, and hence, solve the decision version of the \Stp\ problem with high probability.

\begin{definition}
    For $p\in \Cp$ let $A_p$ be the set of patterns defined in \cref{lem:pat-parity-rep}. We define the operator $\actionst:\Pat\times 4\rightarrow \pow{\CSP}$ as 
    \[\actionst(p,\ell) = A_{\action(p,\ell)},\]
    if $\action(p, \ell)\in \Cp$ or as undefined otherwise. Let $\exacst$ the \emph{exclusive} version of $\actionst$.
\end{definition}
\begin{definition}

    For $x\in V(\syntaxtree), \budget \in [n], \weight \in [n\cdot \W], \actionweight \in [|V(\syntaxtree)|\cdot D]$,
    we define the families $\bas_x\ind{\budget, \weight, \actionweight}$ as follows:

    For an introduce node $i(v)$, let $c = \weightf(v)$. We set $\bas_x\ind{1,c,\actionf(x, i)}= \action\big(\pat_x(\{v\}),i\big)$ for $i\in[2]$. We set $\bas_x\ind{0,0,0}$ to $\{[0]\}$ if $v \notin \term$, and to $\emptyset$ otherwise. For all other values of $\budget, \weight$ and $\actionweight$, we set $\bas_x\ind{\budget,\weight}= \emptyset$.
    For a relabel node $\mu_x = \relabel{i}{j}(\mu_{x'})$, we set
    \[
        \bas_x\ind{\budget,\weight,\actionweight}=  \exrel{\big(\bas_{x'}\ind{\budget,\weight,\actionweight}\big)}{i}{j}
    \]
    For a union node $\mu_x = \mu_{x_1}\union \mu_{x_2}$, we set
    \[
        \bas_x\ind{\budget,\weight,\actionweight} =
        \bigdelta\limits_{\substack{b_1+b_2 = \budget\\c_1+c_2=\weight\\d_1+d_2=\actionweight}}\big(\bas_{x_1}\ind{b_1,c_1,d_1}\exunion\bas_{x_2}\ind{b_2, c_2,d_2}\big).
    \]
    Finally, for a join node $\mu_x = \add{i}{j}(\mu_{x'})$, we define the families
    \[
        \bas'_x\ind{\budget,\weight,\actionweight}= \exadd_{i,j}\big(\bas_{x'}\ind{\budget,\weight,\actionweight}\big),
    \]
    we define the family 
    \[
        \bas_x\ind{\budget,\weight, \actionweight}= \bigdelta\limits_{\ell\in[4]}
            \exacst\big(\bas'_{x}\ind{\budget,\weight,\actionweight - \actionf(x, \ell)}, \ell\big).
    \]
\end{definition}
\begin{lemma}\label{lem:bas-parity-rep-drep}
    For each $x\in V(\syntaxtree)$ and all values of $\budget, \weight, \actionweight$ it holds that $\bas_x\allind \subseteq \CSP$ and that it parity-represents $\drep_x\allind$.
\end{lemma}
\begin{proof}
    It follows from the definition of $\actionst$ and from \cref{obs:cs-closed-ops} that $\bas_x\allind \subseteq \CSP$. For the other part of the lemma, we prove the claim by induction over $\syntaxtree$. For an introduce node $x$, it holds that $\bas_x\allind = \drep_x\allind$. For a relabel or union node, it follows from the induction hypothesis that $\bas_{x'}\ind{b,c,d}$ parity-represents $\drep_{x'}\ind{b,c,d}$ for each child $x'$ of $x$ and all values of $b$, $c$, and $d$. It follows from \cref{lem:ops-preserve-parity-rep} and \cref{lem:mutex-preserve-parity-rep} that $\bas_x\allind$ parity-represents $\drep_x\allind$.

    Finally, for a join node, it holds by \cref{lem:patadd-preserve-parity-rep} that $\bas'_x\ind{\budget, \weight, d}$ parity-represents $\drep'_x\ind{\budget, \weight, d}$, for all values of $d$. It holds by \cref{lem:pat-parity-rep} that $\actionst(p, i)$ parity-represents $\action(p, i)$. Hence, it holds by \cref{lem:ac-preserve-parity-rep} and \cref{lem:mutex-preserve-parity-rep} that $\bas_x\allind$ parity-represents $\drep_x\allind$.
\end{proof}
\begin{corollary}\label{cor:bas-if-drep}
    It holds for all values $\budget, \weight$ and $\actionweight$ that 
    \[[0]\in \bas_r\allind \iff [0] \in \drep_r\allind.\]
\end{corollary}
\begin{proof}
    Since $[0]$ is the only complete pattern consistent with $[0]$, it holds that
    \begin{align}
        [0]\in \bas_r\allind
        &\iff |\{p\in \bas_r\allind\colon p \sim [0]\}|\bquiv 1\\
        &\iff |\{p\in \drep_r\allind\colon p\sim [0]\}|\bquiv 1\\
        &\iff [0]\in \drep_r\allind,\\
    \end{align}
    where the second equivalence follows from \cref{lem:bas-parity-rep-drep}.
\end{proof}
\section{Algorithm}\label{sec:algo}
In this section, we define the final dynamic programming table that computes $\bas_x\ind{\budget, \weight, \actionweight}$ for $x\in V(\syntaxtree)$ and all values $\budget, \weight$ and $\actionweight$ in time $O^*(3^{k})$.

\begin{definition}
    For all $x\in V(\syntaxtree)$ and all values $\budget \in [k]_0, \weight \in [|V|\cdot \W]_0$ and $\actionweight \in [|V(\syntaxtree)|\cdot \D]_0$, we define the vectors $T = T\ind{x, \budget, \weight, \actionweight} \in \{0,1\}^{\CSP}$ as follows:

\begin{itemize}
\item Introduce node $i(v)$. Let $p = [0, i]$ if $v\neq v_0$, and $p = [0i]$ otherwise.
First we initialize $T[x,\budget, \weight, \actionweight] = \overline 0$ for all values $\budget, \weight, \actionweight$. Now we set 
\[T[x, 0, 0, 0][[0]] = [v \notin \term],\]
and for $\ell \in [2]$ we set
\[T[x, 1, \weightf(v), \actionf(x, \ell)][\action(p, \ell)] = 1.\]
\item Relabel node $\mu_x = \relabel{i}{j}(\mu_{x'})$, where $x'$ is the child of $x$ in $\syntaxtree$. We define
\[T\ind{x, \budget, \weight, \actionweight}[p] = \sum\limits_{\substack{q\in \CSP\\ q_{i\rightarrow j}=p}}T\ind{x', \budget, \weight, \actionweight}[q].\]
In order to compute $T\ind{x,\budget, \weight,\actionweight}$, we first initialize it to $\overline 0$. After that, we iterate over all patterns $q \in \CSP$, and we add $T\ind{x', \budget, \weight, \actionweight}[q]$ to $T\ind{x, \budget, \weight, \actionweight}[{q}_{i\rightarrow j}]$.
\item Join node $\mu_x = \add{i}{j}(\mu_{x'})$. For all values $\budget, \weight$, and $\actionweight$, we define the vector $T'\nodeind \in \{0,1\}^{\Pat}$ over all patterns as
\[T'\nodeind[p]=\sum\limits_{\substack{q\in\CSP,\\\patadd_{ij}q=p}}T\ind{x',\budget,\weight,\actionweight}[q],\]
for all patterns $p\in \Pat$. Now for all $p\in \CSP$ we set 
\[T\nodeind[p]= \sum\limits_{\ell\in[4]}
\sum\limits_{\substack{q\in\Pat,\\p\in\actionst(q, \ell)}} T'\ind{x,\budget,\weight,\actionweight - \actionf(x, \ell)}[q].\]
\item Union node $\mu_x = \mu_{x_1}\union \mu_{x_2}$. For all $p\in \CSP$, we define
\[T\nodeind[p] = \sum\limits_{\substack{
    b_1+b_2=\budget\\c_1+c_2=\weight\\d_1+d_2=\actionweight}}
    \sum\limits_{\substack{p_1,p_2\in\CSP\\p_1\punion p_2=p}}
T[x_1, b_1, c_1, d_1][p_1]\cdot T[x_2, b_2, c_2, d_2][p_2].\]
\end{itemize}    
\end{definition}

\begin{lemma}\label{lem:compute-add-node}
    Given a join node $x\in V(\syntaxtree)$, and all tables $T[x', \budget, \weight, \actionweight]$ for all values of $\budget, \weight$ and $\actionweight$, where $x'$ is the child of $x$ in $\syntaxtree$. Then all values $T\nodeind$ can be computed in time $\ostar(3^k)$.
\end{lemma}
\begin{proof}
    First we initialize all vectors $T\nodeind$ to $\overline 0$. After that, we iterate over all values $\budget,\weight, \actionweight$, and for each, we iterate over all patterns $r \in\CSP$. Let $q = \patadd_{i,j}r$. For each value $\ell \in [4]$, and for each pattern $p\in \actionst(q, \ell)$ we add $T\ind{x',\budget,\weight,\actionweight}[q]$ to $T[x, \budget, \weight, \actionweight + \actionf(x, \ell)][p]$.

    It holds that the number of values $\budget, \weight$ and $\actionweight$ we iterate over is bounded by a polynomial in $n$. For each, we iterate over all $CS$-patterns.
    Since it holds by \cref{lem:num-of-cs-pats} that $|\CSP|\leq 3^k$, it suffices to show that we spend at most polynomial time in $k$ for each such pattern. Clearly, we can compute $\patadd_{i,j}p$ and $\actionst(p, \ell)$ in polynomial time in $k$. Since $|\actionst(p, \ell)|\leq 3$, for each pattern we iterate over, and for each $\ell\in[4]$, we update at most three values in the tables $T\nodeind$. Since we operate over $\bin$, we can assume that we can apply each such update in constant time. Hence the running time can be bounded in $\ostar(3^k)$. It holds that 
    \begin{align*}
        T\nodeind[p]
        &=
        \sum\limits_{\substack{r\in\CSP}}
        \sum\limits_{\ell \in [4]}[p\in \actionst(\patadd_{i,j}r)]\cdot T\ind{x', \budget, \weight, \actionweight - \actionf(x, \ell)}[r]\\
        &=\sum\limits_{\ell \in [4]} \sum\limits_{\substack{q\in\Pat\\ p\in\actionst(q, \ell)}}
        \sum\limits_{\substack{r\in\CSP\\\patadd_{i,j}r = q}}
        T\ind{x',\budget,\weight,\actionweight-\actionf(x,\ell)}[r]\\
        &=\sum\limits_{\ell \in [4]} \sum\limits_{\substack{q\in\Pat\\ p\in\actionst(q, \ell)}}
        T'\ind{x,\budget,\weight,\actionweight-\actionf(x,\ell)}[q].
    \end{align*}
\end{proof}

Although computing the tables $T\nodeind$ for a union node $x$ in the naive way yields a running time polynomial in $|\CSP|^2$ which exceeds the bound we seek by far, we make use of a result by Hegerfeld and Kratsch \cite{DBLP:conf/esa/HegerfeldK23} that allows us to compute convolutions over a lattice (called the $\lor$-product) more efficiently. This result is based on another result by Björklund et al. \cite{BjorklundHKKNP16}, and allows us to compute the tables $T\nodeind$ in polynomial time in $|\CSP|$.
\begin{definition}
    Given a lattice $(\mathcal{L}, \preceq)$, and two tables $A,B:\mathcal{L}\rightarrow \mathbb{F}$ for some field $\mathbb{F}$, we define the \emph{join-product} (or $\lor$-product) $A\otimes_{\mathcal{L}} B$ as follows: for each $x\in\mathcal{L}$ we define
    \[A\otimes_{\mathcal{L}} B (x)= \sum\limits_{\substack{y,z\in\mathcal{L}\\y\lor z = x}} A(y)\cdot B(z). \]
\end{definition}
\begin{definition}
    We define a new ordering $\csleq$ over $\CSP$ where $p_1\csleq p_2$, if for each $i \in U$ it holds that $i\in Z_{p_1}$ if $i\in Z_{p_2}$ and $i\in \lbs(p_1)$ if $i \in \lbs(p_2)$.
\end{definition}
\begin{observation}\label{obs:union-is-convolution}
    The join operation over the lattice $(\CSP, \csleq)$ is given by $p\lor q = r$, where $Z_r = Z_p\cup Z_q$ and $\lbs(r)=\lbs(p)\cup \lbs(q)$. This corresponds exactly to the union operation over $CS$-patterns $r = p \punion q$. Hence it holds for a union node that
    \[T\nodeind = \sum\limits_{\substack{
    b_1+b_2=\budget\\c_1+c_2=\weight\\d_1+d_2=\actionweight}}
T[x', b_1, c_1, d_1] \otimes_{\CSP} T[x'', b_2, c_2, d_2].\]
\end{observation}

\begin{definition}
    Given a lattice $(\mathcal{L}, \preceq)$. We say that an element $x\in\mathcal{L}$ is \emph{join-irreducible} if $x = a\lor b$ implies that $x=a$ or $x=b$. We denote by $\mathcal{L}_{\lor}$ the set of all join-irreducible elements of $\mathcal{L}$. Note that $\hat 0 \in \mathcal{L}$ always holds.

    We say that a lattice $(\mathcal{L}, \preceq)$ is given in the \emph{join-representation} if the elements of $\mathcal{L}$ are represented as $O(\log |\mathcal{L}|)$-bit strings, together with the set of join-irreducible elements $\mathcal{L}_{\lor} \subseteq \mathcal{L}$, and an algorithm $\mathcal{A}_{\mathcal{L}}$ that computes the join $a\lor v$ given an element $a\in\mathcal{L}$ and a join-irreducible element $x\in\mathcal{L}_{\lor}$.
\end{definition}
First we state the result by Björklund et al. that serves as a basis for the following result by Hegerfeld and Kratsch.
\begin{theorem}[\cite{BjorklundHKKNP16}]
    Let $(\mathcal{L},\preceq)$ be a finite lattice given in join-representation and $A,B\colon\mathcal{L}\rightarrow \mathbb{F}$ be two tables, where $\mathbb{F}$ is some field. The $\lor$-product $A\otimes_{\mathcal{L}}B$ can be computed in $O(|\mathcal{L}||\mathcal{L}_{\lor}|)$ field operations and calls to algorithm $\mathcal{A}_{\mathcal{L}}$ and further time $O(|\mathcal{L}||\mathcal{L}_{\lor}|^2)$.
\end{theorem}
Based on this theorem, Hegerfeld and Kratsch proved the following result
\begin{corollary}[\cite{DBLP:conf/esa/HegerfeldK23} - Corollary A.10]\label{cor:falko-lattice}
    Let $(\mathcal{L}, \preceq)$ be a finite lattice given in join-representation and $k$ be a natural number. Given two tables $A,B\colon\mathcal{L}^k\rightarrow \mathbb{Z}_2$, the $\lor$-product $A\otimes_{\mathcal{L}^k}B$ in $\mathcal{L}^k$ can be computed in time $O(k^2|\mathcal{L}|^{k+2})$ and $O(k|\mathcal{L}|^{k+1})$ calls to algorithm $\mathcal{A}_{\mathcal{L}}$.
\end{corollary}

Now we show that the lattice $(\CSP, \csleq)$ is isomorphic to a simpler lattice that can be written as an exponent of a very basic lattice.
\begin{definition}
    For $k\in\mathbb{N}$, we define the family of tuples $F_k = \{0,1,2\}^k$ (we call each of these tuples a \emph{state function}). Let $\tau$ be the bijective mapping  between $\CSP(U)$ and $F_k$, where we define $\tau(p) = (\ell_1, \dots \ell_k)$, such that 
    \[
        \ell_i=
        \begin{cases}
            2 &\colon i \in Z_p,\\
            1 &\colon i \in \lbs(p)\setminus Z_p,\\
            0 &\colon \text{otherwise.}
        \end{cases}
        \]
\end{definition}
For a tuple $t$, let $t_i$ be the $i$th element in this tuple.
\begin{definition}
    We define the ordering $\preceq$ over $F_k$, where for two tuples $a,b \in F_k$ it holds $a \preceq b$ if $a_i \leq b_i$ for all $i\in [k]$. Hence, the join operation $c = a\lor b$ over the lattice $(F_k, \preceq)$ is given by $c_i = \max(a_i, b_i)$ for all $i\in[k]$.
\end{definition}
\begin{lemma}\label{lem:tau-is-isomorphism}
    The mapping $\tau$ is an isomorphism between the lattice $(\CSP(U), \csleq)$ and the lattice $(F_k, \preceq)$.
\end{lemma}
\begin{proof}
    Let $p, q$ be two complete patterns, and $t= \tau(p), t'=\tau(q)$. 
    It holds that $p\csleq q$ if and only if for each $i\in [k]$ it holds that $i \in Z_p$ implies that $i \in Z_q$ and $i\in \lbs(p)$ implies that $i\in\lbs(q)$. This is the case if and only if for each $i \in [k]$ it holds that $t_i = 2$ implies that $t'_i = 2$ and $t_i = 1$ implies that $t'_i \geq 1$, which is the case if and only if $t \preceq t'$.
\end{proof}
\begin{observation}\label{obs:Fk-is-power-of-012}
    Let $\mathcal{L}$ be the lattice $(\{0,1,2\}, \leq)$. It holds that $(F_u, \preceq) = \mathcal{L}^k$.
\end{observation}

\begin{corollary}\label{cor:compute-union-node}
    For a union node $x$, and some values $\budget, \weight, \actionweight$, the table $T\nodeind$ can be computed in time $O^*(3^k)$.
\end{corollary}
\begin{proof}
    First we initialize $T\nodeind = \overline 0$.
    For all non-negative integers $b_1,b_2,c_1,c_2$ and $d_1, d_2$, such that $b_1+b_2 = \budget$, $c_1+c_2 = \weight$, and $d_1 + d_2 = \actionweight$,
    using the bijection $\tau$, we build the tables $S_1, S_2 : F_k \rightarrow \{0,1\}$, where $S_1[\tau(p)]=T\ind{x_1,b_1,c_1,d_1}[p]$, and $S_2[\tau(p)] = T\ind{x_2,b_2,c_2,d_2}[p]$ for all $p\in \CSP$.
    Let $S = S_1\otimes_{F_k}S_2$. Since we can compute the join of two elements in $F_k$ in polynomial time in $k$, it holds by \cref{cor:falko-lattice} and \cref{obs:Fk-is-power-of-012} that $S$ can be computed in polynomial time in $|\mathcal{L}|^k = 3^k$, where $\mathcal{L} = \{0, 1, 2\}$.

    We define the table $T'\colon\CSP\rightarrow \mathbb{Z}_2$ as $T'[p] = S[\tau(p)]$.
    By \cref{lem:tau-is-isomorphism}, it holds that $\tau$ is an isomorphism. Hence, it holds for $p\in \CSP$ that
    \begin{align*}
        T'[p]
        &= S[\tau(p)]\\ 
        &= S_1\otimes_{F_k}S_2[\tau(p)]\\
        &=\left(T\ind{x_1,b_1,c_1,d_1}\otimes_{\CSP}T\ind{x_2,b_2,c_2,d_2}\right)[p].
    \end{align*}
    We add $T'$ to $T\nodeind$. The correctness follows from \cref{obs:union-is-convolution}.
    The number of different values $b_1, b_2, c_1, c_2, d_1$ and $d_2$ we iterate over is bounded by polynomial in $n$, and for each we spend time polynomial in $3^k$. Hence, the whole computation runs in time $\ostar(3^k)$.
\end{proof}

\begin{lemma}\label{lem:algo-running-time}
    The families $T\nodeind$ for all $x\in V(\syntaxtree)$ and all values $\budget, \weight,\actionweight$ can be computed in time $\ostar(3^k)$.
\end{lemma}
\begin{proof}
    Since all tables $T\nodeind$ are defined over $\bin$, we assume that 
    we can apply addition and multiplication over elements of these vectors in constant time.
    We also assume that we can apply the operations in \cref{def:patops} and \cref{def:patadd} over patterns in time polynomial in $k$.
    Since $|V(\syntaxtree)|$ is polynomial in $n$, it suffices to show for each node $x\in V(\syntaxtree)$ that it takes time $\ostar(3^k)$ to compute all tables $T\nodeind$. For an introduce node this is clearly the case. For a relabel node, we iterate over all values of $\budget \in [n], \weight \in [\W\cdot n], \actionweight \in [\D\cdot |V(\syntaxtree)|]$, and over all patterns in $\CSP$ and we process each such pattern in polynomial time in $k$. Since the three former values are bounded by a polynomial in $n$, and since it holds by \cref{lem:num-of-cs-pats} that $|\CSP| = 3^k$, it follows that we process a relabel node in time $\ostar(3^k)$.

    For a join node this follows from \cref{lem:compute-add-node}. For a union node, we iterate over polynomially many different values $\budget, \weight, \actionweight$, and for each we compute $T\nodeind$ by \cref{cor:compute-union-node} in time $\ostar(3^k)$.
\end{proof}
\begin{lemma}\label{lem:algo-counts-par-rep}
    It holds for $x\in V(\syntaxtree)$, and all values $\budget, \weight, \actionweight$, and a $CS$-pattern $p\in \CSP$ that 
    \[T\nodeind[p] = \left[p\in \bas_x\allind\right].\]
\end{lemma}
\begin{proof}
    We prove the claim by induction over $\syntaxtree$. For an introduce node, the claim holds clearly. For a relabel node, it holds that 
    \begin{align*}
    p\in \bas_x\allind 
    &\iff |\{p'\in C_{x'}\allind\colon p'_{i\rightarrow j} = p\}| \bquiv 1\\
    &\iff \sum\limits_{p'_{i\rightarrow j} = p}T[x', \budget, \weight, \actionweight][p'] \bquiv 1\\
    &\iff T\nodeind[p] \bquiv 1,
    \end{align*}
    where the second equivalence holds due to the induction hypothesis.
    For a union node, it holds that 
    \begin{align*}
    p\in \bas_x\allind 
    &\iff \Big|\bigdelta\limits_{\substack{b_1+b_2=\budget\\c_1+c_2=\weight\\d_1+d_2=\actionweight}}\big\{(p_1, p_2)\in \bas_{x_1}\ind{b_1,c_1, d_1}\times \bas_{x_2}\ind{b_2,c_2,d_2}\colon p_1\punion p_2 = p\big\}\Big| \bquiv 1\\
    &\iff \sum\limits_{\substack{
        b_1+b_2=\budget\\c_1+c_2=\weight\\d_1+d_2=\actionweight\\p_1\punion p_2 = p}}
    T[x', b_1, c_1, d_1][p_1]\cdot T[x'', b_2, c_2, d_2][p_2] \bquiv 1\\
    &\iff T\nodeind[p] \bquiv 1.
    \end{align*}
    For a join node, it holds for $p\in \Pat$ that 
    \begin{equation}\label{eq:Tprime}
        \begin{split}            
        p\in \bas'_x\allind 
        &\iff |\{p'\in \bas_{x'}\allind\colon \patadd_{i,j}p' = p\}| \bquiv 1\\
        &\iff \sum\limits_{\patadd_{i,j}p' = p}T[x', \budget, \weight, \actionweight][p'] \bquiv 1\\
        &\iff T'\nodeind[p] \bquiv 1.
    \end{split}
\end{equation}
    Hence, it holds that
    \begin{align*}
        p \in \bas_x\allind
        &\iff p \in \bigdelta\limits_{\ell\in[4]}
        \exacst\big(\bas'_{x}\ind{\budget,\weight,\actionweight - \actionf(x, \ell)}, \ell\big)\\
        &\iff \sum\limits_{\ell\in[4]}|\{q\in\bas'_{x}\ind{\budget,\weight,\actionweight - \actionf(x, \ell)}\colon p\in\actionst(q, \ell)\}|\bquiv 1\\
        &\iff \sum\limits_{\ell\in[4]}
        \sum\limits_{p\in\actionst(q, \ell)}T'\ind{\budget,\weight,\actionweight - \actionf(x, \ell)}[q]\bquiv 1\\
        &\iff T\nodeind \bquiv 1,
    \end{align*}
    where the third congruence follows from (\ref{eq:Tprime}).
\end{proof}
\begin{corollary}
    It holds for all values $\budget, \weight$ and $\actionweight$ that 
    \[T\ind{r,\budget,\weight,\actionweight}\left[[0]\right]\bquiv 1 \iff [0] \in \drep_r\allind.\]
\end{corollary}
\begin{proof}
    This follows directly from \cref{lem:algo-counts-par-rep} and \cref{cor:bas-if-drep}.
\end{proof}

\begin{proof}[Proof (\cref{theorem:intro:mainresult})]
The algorithm first computes all families $T\ind{x,b,\weight,\actionweight}$ for all $x\in V(\syntaxtree)$ and all values of $b, \weight, \actionweight$. The algorithm states that there exists a Steiner tree of size $\budget$, if and only if there exist two values $\weight, \actionweight$ with $T\ind{r, \budget, \weight, \actionweight}\left[[0]\right] = 1$.

By \cref{lem:algo-running-time}, we can compute all families $T\ind{x,b,\weight,\actionweight}$ in time $\ostar(3^k)$. 
By \cref{lem:algo-counts-par-rep}, it holds for all values $\budget,\weight,\actionweight$ that
\[T\ind{r,\budget,\weight,\actionweight}\left[[0]\right]\bquiv 1 \iff [0] \in \drep_r\allind.\]
By \cref{lem:no-sol-no-drep} it holds that if no Steiner tree of size $\budget$ exists, then the algorithm will always return false. On the other hand if a Steiner tree of size $\budget$ exists, by choosing $\W = (2+\sqrt 2) |V|$, and $\D = 4(2+\sqrt 2) |V(\syntaxtree)|$, and by choosing both function $\weightf\in [\W]^{V}$, and $\actionf\in [\D]^{V(\syntaxtree)\times [4]}$ uniformly and independently at random, it holds by \cref{lem:iso-total-prob} that the algorithm will correctly decide the existence of such a tree with probability at least one half.
\end{proof}

\section{Conclusion}\label{section:conclusion}

We design a one-sided Monte Carlo algorithm for the \textsc{Steiner Tree[$\cw$]} problem that given a graph $G=(V,E)$, a set $T\subseteq V$ of terminals, and a $k$-clique-expression computes in time $3^k\cdot n^{\Oh(1)}$ the minimum number of vertices over all connected subgraphs $H$ of $G$ that span $T$. Due to an existing lower bound for \textsc{Steiner Tree[$\pw$]} this establishes (modulo SETH) that the basis for the correct parameter dependence for \textsc{Steiner Tree[$\cw$]} is $3$. This answers an open problem of Hegerfeld and Kratsch~\cite{DBLP:conf/esa/HegerfeldK23} and settles the complexity of arguably \emph{the} prototypical connectivity problem relative to clique-width. 

One technical contribution is the identification of so-called \emph{complete patterns} for dealing with connectivity in the setting of clique-width and labeled graphs, as these are representative for all patterns. This alone may be sufficient to settle the complexity of further connectivity problems such as \textsc{Connected Odd Cycle Transversal[$\cw$]} (cf.~\cite{DBLP:conf/esa/HegerfeldK23}). Our second technical contribution, to \emph{isolate a representative solution}, is likely more broadly applicable: By giving weights to the possible actions in a DP, which in particular may create new representatives for a given partial solution, we are able to isolate among representatives. Thus, if there is also a unique solution to begin with (as can be ensured w.h.p.), then we also have a unique representative (again w.h.p.). This permits the combined use of both representative sets as well as low-rank factorization-based approaches for matrices of low GF(2)-rank and may be instrumental in overcoming rank-submatrix gaps beyond the immediate realm of connectivity problems.

Apart from \textsc{Connected Odd Cycle Transversal[$\cw$]}, both \textsc{Feedback Vertex Set[$\cw$]} and \textsc{Connected Feedback Vertex Set[$\cw$]} are good problems whose exact complexity to aim for. In particular, \textsc{Feedback Vertex Set[$\cw$]} brings a known difficulty in counting the number of edges outside the solution to enforce the acyclicity requirement (see Hegerfeld and Kratsch~\cite{DBLP:conf/esa/HegerfeldK23} and Bergougnoux and Kant\'e~\cite{DBLP:journals/tcs/BergougnouxK19,DBLP:journals/siamdm/BergougnouxK21}).

\bibliography{ref}
\end{document}